\documentclass[10pt, reqno]{amsart}
\usepackage{fullpage}
\usepackage{amsmath, amsthm}
\usepackage{amscd}
\usepackage{enumerate}
\usepackage{float,amsmath,amssymb,mathrsfs,bm,multirow,graphics}
\usepackage[dvips]{graphicx}
\usepackage[percent]{overpic}
\usepackage[numbers,sort&compress]{natbib}
\usepackage{xcolor}
\usepackage{todonotes}
\usepackage{amsaddr}
\usepackage{mathtools}
\usepackage{subcaption}
\usepackage{mathtools}
\usepackage{stackengine} 
 
\usepackage{tikz}
\usepackage{tikz-cd}
\usepackage{tqft}
\usetikzlibrary{decorations.pathreplacing,decorations.markings}

\newtheorem{theorem}{Theorem}
\newtheorem{proposition}{Proposition}[section]
\newtheorem{lemma}[proposition]{Lemma}

\newtheorem{assumption}[proposition]{Assumption}

\newcommand{\im}{\text{\upshape Im} \,}

 \newcommand\ben{\begin{equation*}}
 \newcommand\ebn{\end{equation*}}
 \newcommand\beq{\begin{equation}}
 \newcommand\eeq{\end{equation}}
 \newcommand\lb{\left(}
  \newcommand\rb{\right)} 

\numberwithin{equation}{section}

\allowdisplaybreaks[1]

\usepackage[colorlinks=true]{hyperref}
\hypersetup{urlcolor=blue, citecolor=red, linkcolor=blue}
\title{The Virasoro fusion kernel and Ruijsenaars' hypergeometric function}
\author{Julien Roussillon}
\address{Department of Mathematics, KTH Royal Institute of Technology, \\ 100 44 Stockholm, Sweden.}
\email{julienro@kth.se}
\begin{document}
\begin{abstract}
We show that the Virasoro fusion kernel is equal to Ruijsenaars' hypergeometric function up to normalization. More precisely, we prove that the Virasoro fusion kernel is a joint eigenfunction of four difference operators. We find a renormalized version of this kernel for which the four difference operators are mapped to four versions of the quantum relativistic hyperbolic Calogero-Moser Hamiltonian tied with the root system $BC_1$. We consequently prove that the renormalized Virasoro fusion kernel and the corresponding quantum eigenfunction, the (renormalized) Ruijsenaars hypergeometric function, are equal.
\end{abstract}
\maketitle
\tableofcontents

\section{Introduction}

Two-dimensional conformal field theories (CFTs) have been intensively studied since the seminal work of Belavin, Polyakov and Zamolodchikov in \cite{BPZ}. In addition to their broad range of physical applications in for example condensed matter theory and string theory, they possess rich underlying mathematical structures. 

The infinite-dimensional symmetry algebra of 2D CFTs determines the structure of their correlation functions. The operator product expansion makes it possible to decompose $N$-point correlators on the Riemann sphere into combinations of three-point structure constants and universal quantities called \textit{Virasoro conformal blocks}. In particular, there are three ways to decompose the four-point correlation function. The universal parts of each decomposition are called $s$-, $t$-, and $u$-channel conformal blocks. Since each decomposition must lead to the same four-point function, the three kinds of conformal blocks are related by duality transformations.

In particular, the $s$- and $t$-channel conformal blocks are related by an integral transform called fusion transformation. Such a transformation was conjectured in \cite{pt1} and established in \cite{T03}. The corresponding kernel, the \textit{Virasoro fusion kernel}, was constructed by Ponsot and Teschner in \cite{pt1,pt2} and revisited in \cite{TV14} as \textit{$b$-6j symbols} for the modular double of $\mathcal{U}_q(sl_2(\mathbb{R}))$. This terminology comes from the fact that $b$ is associated to two unimodular parameters $q=e^{i\pi b^2}$ and $\tilde{q}=e^{i\pi b^{-2}}$. The Virasoro fusion kernel also appears in quantum Teichm\"uller theory \cite{NT}, in 3d supersymmetric gauge theories \cite{TV14,TY}, and recent physical applications were found in \cite{CBMP,CL,K}.

From the conformal blocks viewpoint, the Virasoro fusion kernel is denoted $F\lb b,\boldsymbol\theta,\sigma_s,\sigma_t \rb$. Here, $b$ is used to characterize the central charge of the CFT as $c=1+6(b+b^{-1})^2$, while $\boldsymbol\theta$ is a set of four external conformal dimensions associated to the fields entering the four-point correlation function. Finally, $\sigma_s$ and $\sigma_t$ are two internal conformal dimensions associated to the s- and t-channel decompositions of the four-point function.

The Ponsot-Teschner formula \cite{pt1} for $F$ is recalled in \eqref{fusion01}. It is given by a contour integral whose integrand involves a special function $s_b(z)$. The $s_b$-function appears in several different contexts. It was introduced under the names of "Quantum Dilogarithm function", "Hyperbolic Gamma function" and "Quantum Exponential function" in \cite{FK,W2000,R1996}, respectively\footnote{The hyperbolic gamma function first appeared in \cite{R1994} but was named as such in \cite{R1996}.}. In particular, $F$ is proportional to a \textit{hyperbolic Barnes integral} which is a hyperbolic generalization of the Barnes representation for the Gauss hypergeometric function. Such an integral was constructed in \cite{BRS} as a degeneration limit of the hyperbolic hypergeometric function, and it also arises as a limit of Spiridonov's elliptic hypergeometric $V$-function \cite{TV14,Sp}. Finally, the Virasoro fusion kernel is a joint eigenfunction of four difference operators \cite{pt1}.

On the other hand, Calogero-Moser (CM) models are systems of particles living on the real line or the circle and interacting with a rational, trigonometric/hyperbolic or elliptic potential. Their relativistic deformations were found by Ruijsenaars and Schneider at the classical level \cite{RS}. Quantum relativistic CM systems tied with the root system $A_{N-1}$ were found by Ruijsenaars \cite{R87}, and their $BC_N$ generalization by Van Diejen \cite{V94}. Therefore we will refer to quantum relativistic $BC_N$ CM systems as quantum Ruijsenaars-van Diejen (RvD) systems.

The rank $N$ quantum trigonometric RvD system is solved by the Koornwinder N-variable polynomials \cite{Koo}. However, in the hyperbolic case the eigenfunctions are non-polynomial and are only known in the rank one. The corresponding eigenfunction, the \textit{Ruijsenaars hypergeometric function}, was introduced in \cite{R1994} and studied in greater detail in \cite{R1999,R2003,R2003bis}. This function is denoted $R(a_-,a_+,\boldsymbol{\gamma},v,\hat{v})$; here, $a_-$ and $a_+$ are associated to two unimodular quantum deformation parameters $q=e^{i\pi a_-/a_+}$ and $\tilde{q}=e^{i \pi a_+/a_-}$, while $\boldsymbol\gamma$ is a set of four external couplings constants. Finally,  $v$ and $\hat v$ are viewed as geometric and spectral variables, respectively. 

Of particular importance is the renormalized $R$-function, denoted $R_\text{ren}(a_-,a_+,\boldsymbol{\gamma},v,\hat{v})$, which was defined in \cite{R1999}. The definition \eqref{fonctionR} and the properties of the function $R_\text{ren}$ resemble those of the Virasoro fusion kernel. First of all, $R_\text{ren}$ is defined as a contour integral whose integrand involves the hyperbolic gamma function $G(a_-,a_+,z)$. The functions $G$ and $s_b$ are simply related by 
\begin{equation*} s_b(z) = G(b,b^{-1},z). \end{equation*} 
$R_\text{ren}$ is also proportional to a hyperbolic Barnes integral \cite{BRS}. Moreover, it is a joint eigenfunction of four difference operators which are four versions of the rank one quantum hyperbolic RvD Hamiltonian. Let us finally mention that the function $R_\text{ren}$ was related to the modular double of $\mathcal{U}_q(sl_2(\mathbb{R}))$ in \cite{B}.

The discussion above suggests  that the Virasoro fusion kernel and the Ruijsenaars hypergeometric function are closely related. However there are at least three obstacles to finding an exact relation between them:
\begin{enumerate}
\item $F\lb b,\boldsymbol\theta,\sigma_s,\sigma_t \rb$ and $R_\text{ren}(a_-,a_+,\boldsymbol{\gamma},v,\hat{v})$ are expressed in terms of different parameters.
\item There is a normalization ambiguity in the definitions of $F$ and $R_\text{ren}$. For instance, the definition of $F$ depends on the normalization of the s- and t-channel conformal blocks.
\item Both $F$ and $R_\text{ren}$ possess several equivalent representations because the hyperbolic Barnes integral satisfies various nontrivial identities.
\end{enumerate}

The aim of this article is to overcome these obstacles and to show that the Virasoro fusion kernel and the Ruijsenaars hypergeometric function are the same function up to normalization. The identification procedure can be summarized as follows. We prove in Propositions \ref{prop1} and \ref{prop4p2} that the Virasoro fusion kernel is a joint eigenfunction of four difference operators. We provide the Ruijsenaars/CFT parameters identification in Section \ref{section5p1}. We define in \eqref{tildeF} a renormalized version $F_\text{ren}$ of $F$ and we show in Proposition \ref{prop2} that both $R_\text{ren}$ and $F_\text{ren}$ satisfy the same four difference equations. Therefore, the two functions are proportional. We finally show in Theorem \ref{theoremRF} that they are actually equal.


%
%
%

\subsection{Organization of the paper} The Ruijsenaars hypergeometric function and various of its properties are recalled in Section \ref{section2}. Section \ref{section3} reviews some properties of the four-point Virasoro conformal blocks and introduces the Virasoro fusion kernel in this context. In Section \ref{section4} we study various symmetry and eigenfunction properties of $F$. Finally, our main result --- the identification of $F$ with $R_\text{ren}$ --- is presented in Section \ref{section5}.

\section{Ruijsenaars' hypergeometric function}\label{section2}

This section is a brief overview of the (renormalized) Ruijsenaars hypergeometric function following \cite{R1999}. We recall several of its properties which we will need to relate it to the Virasoro fusion kernel.
\subsection{Definition}
Define two sets of external couplings constants $\boldsymbol\gamma$ and $\hat{\boldsymbol\gamma}$ by
\beq\label{gamma}
\boldsymbol\gamma=\begin{pmatrix}\gamma_0\\ \gamma_1 \\ \gamma_2\\ \gamma_3\end{pmatrix}, \qquad \hat{\boldsymbol\gamma}=\begin{pmatrix}\hat\gamma_0\\ \hat\gamma_1 \\ \hat\gamma_2\\ \hat\gamma_3\end{pmatrix}= J \boldsymbol\gamma,\eeq
where $J$ is a matrix satisfying $J^2=\mathcal{I}_4$ and is defined by
\beq\label{J}
J=\frac{1}{2} \begin{pmatrix}1 & 1 & 1 & 1 \\ 1 & 1 & -1 &-1\\1&-1&1&-1\\1&-1&-1&1  \end{pmatrix}.
\eeq
The renormalized $R$-function denoted $R_\text{ren}(a_-,a_+,\boldsymbol{\gamma},v,\hat{v})$ is given by a contour integral whose integrand involves the hyperbolic gamma function $G(a_-,a_+,z)$. The definition of $G$ and its properties can be found in Appendix \ref{appendixA}. Its dependence on $a_-$ and $a_+$ will be omitted for simplicity. Following \cite{R1999}, the function $R_\text{ren}$ is defined for $(a_-,a_+,\boldsymbol\gamma) \in \text{RHP}^2 \times \mathbb{C}^4$, where RHP denotes the open right-half plane, such that
\begin{equation}\label{fonctionR}\begin{split}
&R_{\text{ren}}(a_-,a_+,\boldsymbol{\gamma},v,\hat{v})=\frac{1}{\sqrt{a_+ a_-}} \prod_{\epsilon=\pm} \frac{G(\epsilon v-i \gamma_0)}{G(\epsilon \hat{v}+i \hat{\gamma}_0)}\displaystyle \int_{\mathsf{R}} dz \frac{\prod_{\epsilon=\pm}G(z + \epsilon v+i \gamma_0)G(z + \epsilon \hat{v}+i \hat{\gamma}_0)}{G(z+i a)\prod_{j=1}^3 G(z+i \gamma_0+i \gamma_j+i a)},
\end{split}\end{equation}
and where $a=\frac{a_++a_-}{2}$. The function $G$ has sequences of poles and zeros given in \eqref{polesGhyp} and \eqref{zeroGhyp}, respectively. Therefore, for instance when $a_+,a_->0$ the integrand in \eqref{fonctionR} has eight semi-infinite lines of poles, four of them increasing and the other ones decreasing. The integration contour $\mathsf{R}$ runs from $-\infty$ to $+\infty$ separating the upper and lower sequences of poles. Moreover, for fixed values of $a_-,a_+$ and $\boldsymbol\gamma$, the function $R_\text{ren}$ is meromorphic in $v$ and $\hat{v}$. Let us also note that the functions $R$ and $R_\text{ren}$ are related as follows:
\beq
R(a_-,a_+,\boldsymbol{\gamma},v,\hat{v}) = R_\text{ren}(a_-,a_+,\boldsymbol{\gamma},v,\hat{v}) \prod_{j=1}^3 G(i \gamma_0+i\gamma_j+ia).
\eeq
Finally, it was shown in \cite[Proposition 4.20]{BRS}~that the function $R_\text{ren}$ is proportional to a hyperbolic Barnes integral $\mathcal{B}_h(a_-,a_+,\boldsymbol{u})$ whose definition is recalled in Appendix \ref{appendixB}. $\mathcal{B}_h(a_-,a_+,\boldsymbol{u})$ is one possible degeneration of the hyperbolic hypergeometric function and other types of degenerations led the authors of \cite{BRS} to find other representations for $R_\text{ren}$. 
\subsection{Symmetry properties}

The function $R_\text{ren}$ possesses various discrete symmetries which follow from the representation \eqref{fonctionR}. In view of the property \eqref{scale} of the hyperbolic gamma function, $R_\text{ren}$ is scale invariant:
\beq\label{scaleR}
 R_{\text{ren}}(\lambda a_-,\lambda a_+,\lambda \boldsymbol{\gamma},\lambda v,\lambda \hat{v}) = R_{\text{ren}}(a_-,a_+,\boldsymbol{\gamma},v,\hat{v}), \qquad \lambda > 0.
\eeq
Moreover, because of the fact that $G(a_-,a_+,z)=G(a_+,a_-,z)$, the following identity holds:
\beq\label{modulardual}
 R_{\text{ren}}(a_-,a_+,\boldsymbol{\gamma},v,\hat{v})=R_{\text{ren}}(a_+,a_-,\boldsymbol{\gamma},v,\hat{v}) \qquad \text{(modular symmetry).}
 \eeq
It can also be verified that $R_\text{ren}$ satisfies
 \beq \label{selfdual}  R_{\text{ren}}(a_-,a_+,\boldsymbol{\gamma},v,\hat{v})=R_{\text{ren}}(a_-,a_+,\hat{\boldsymbol{\gamma}},\hat{v},v) \qquad \text{(self-duality).}
\eeq
Finally, the function $R_\text{ren}$ is even in $v,\hat{v}$ and is symmetric under permutations of $\gamma_1,\gamma_2,\gamma_3$. It was shown in \cite[Theorem 1.1]{R2003bis} that a similarity transformed version of $R_\text{ren}$ has an extended $D_4$-symmetry in $\boldsymbol\gamma$. 

\subsection{Eigenfunction properties}
Define a translation operator $e^{\pm i a_- \partial_x}$ which formally acts on meromorphic functions $f(x)$ as $e^{\pm i a_- \partial_x} f(x)=f(x\pm i a_-)$. The quantum relativistic Calogero-Moser Hamiltonian tied with the root system $BC_1$ and with a hyperbolic interaction was defined in \cite{V94}. As mentioned in the introduction, it will be referred to as the Ruijsenaars-van Diejen (RvD) Hamiltonian. It is a difference operator of the form 
\begin{equation}\label{hamiltonienR}
H_{RvD}(a_-,a_+,\boldsymbol{\gamma},x)=C(a_-,a_+,\boldsymbol{\gamma},x)e^{-i a_- \partial_x}+C(a_-,a_+,\boldsymbol{\gamma},-x)e^{i a_- \partial_x} + V(a_-,a_+,\boldsymbol{\gamma},x),
\end{equation}
where 
\begin{equation}\label{fonctionC}
C(a_-,a_+,\boldsymbol{\gamma},x)=
-\frac{4\prod_{k=0}^3\operatorname{cosh}{\lb\frac{\pi}{a_+}(x-i\gamma_k-\frac{i a_-}{2})\rb}}{\operatorname{sinh}{\lb\frac{2\pi x}{a_+}\rb}\operatorname{sinh}{\lb\frac{2\pi}{a_+}(x-\frac{i a_-}{2})\rb}}, 
\end{equation}
and where the potential $V$ is defined by
\beq \label{potentialV}
V(a_-,a_+,\boldsymbol{\gamma},x)=-C(a_-,a_+,\boldsymbol{\gamma},x)-C(a_-,a_+,\boldsymbol{\gamma},-x)-2\operatorname{cos}{\lb \frac{\pi}{a_+}\lb \sum_{\mu=0}^3 \gamma_\mu+a_-\rb\rb}.
\eeq
The function $R_\text{ren}$ is a joint eigenfunction of four versions of this Hamiltonian \cite[Theorem 3.1]{R1999}. More precisely, for $(a_-,a_+,\boldsymbol{\gamma},v,\hat{v})\in \text{RHP}^2\times\mathbb{C}^6$ the following eigenvalue equations hold:
\begin{subequations}\label{eigenvaluesR}\begin{align}
\label{differenceR1} & H_{RvD}(a_-,a_+,\boldsymbol{\gamma},v)R_{\text{ren}}(a_-,a_+,\boldsymbol{\gamma},v,\hat{v})=2\operatorname{cosh}{\lb\tfrac{2\pi \hat{v}}{a_+}\rb}R_{\text{ren}}(a_-,a_+,\boldsymbol{\gamma},v,\hat{v}), \\
\label{differenceR2} & H_{RvD}(a_+,a_-,\boldsymbol{\gamma},v)R_{\text{ren}}(a_-,a_+,\boldsymbol{\gamma},v,\hat{v})=2\operatorname{cosh}{\lb\tfrac{2\pi \hat{v}}{a_-}\rb}R_{\text{ren}}(a_-,a_+,\boldsymbol{\gamma},v,\hat{v}), \\
\label{differenceR3} & H_{RvD}(a_-,a_+,\hat{\boldsymbol{\gamma}},\hat{v})R_{\text{ren}}(a_-,a_+,\boldsymbol{\gamma},v,\hat{v})=2\operatorname{cosh}{\lb\tfrac{2\pi v}{a_+}\rb}R_{\text{ren}}(a_-,a_+,\boldsymbol{\gamma},v,\hat{v}), \\
\label{differenceR4} & H_{RvD}(a_+,a_-,\hat{\boldsymbol{\gamma}},\hat{v})R_{\text{ren}}(a_-,a_+,\boldsymbol{\gamma},v,\hat{v})=2\operatorname{cosh}{\lb\tfrac{2\pi v}{a_-}\rb}R_{\text{ren}}(a_-,a_+,\boldsymbol{\gamma},v,\hat{v}). 
\end{align}\end{subequations}

It can be noted that \eqref{differenceR1} implies \eqref{differenceR2}, \eqref{differenceR3} and \eqref{differenceR4} thanks to the properties \eqref{modulardual} and \eqref{selfdual} of $R_\text{ren}$.

Let us finally mention that solving a quantum relativistic integrable system consists of two steps. A joint eigenfunction $J$ first needs to be found by solving a system of difference equations. Then, a Hilbert space transform associated to $J$ has to be constructed. The latter is a quantum counterpart of the action-angle transform for classical integrable systems. In the case of the rank one quantum hyperbolic RvD system, the second step of this program was realized in \cite{R2003} for special values of the couplings.

\section{Four-point Virasoro conformal blocks}\label{section3}

In this section we briefly review the four-point Virasoro conformal blocks and recall how the Virasoro fusion kernel arises in this context. A more complete overview of the subject can be found in \cite{T01,T03,Ribault}. 
 


\subsection{Highest-weight representations of the Virasoro algebra} The Virasoro algebra is the symmetry algebra of two-dimensional conformal field theories. It has generators $L_n$, $n\in \mathbb{Z}$ and relations
\beq \label{virasoro}
[L_n,L_m]=(n-m)L_{n+m} + \frac{c}{12}n(n^2-1)\delta_{n+m,0},\eeq
where $c$ is a central element in the algebra called central charge. Highest weight representations $\mathcal{V}_\theta$ of the Virasoro algebra are generated from vectors $\left|\theta\right>$ which satisfy
\beq
L_0\left|\theta\right> = \Delta(\theta) \left|\theta\right>,\qquad L_n\left|\theta\right> = 0, \quad n>0,
\eeq
where a Liouville type parametrization of $\Delta$ and $c$ is used:
\beq\label{liouvilleparam}
\Delta(x)=\tfrac{Q^2}4+x^2, \qquad c=1+6Q^2,\qquad Q=b+b^{-1}.
\eeq
The representations $\mathcal{V}_\theta$ decompose as follows:
\beq
\mathcal{V}_\theta \simeq \bigoplus_{n\in \mathbb{Z} \geq 0} \mathcal{V}_\theta^{(n)},
\eeq
where $L_0 v = (\Delta(\theta)+n) v$ for all $v \in \mathcal{V}_\theta^{(n)}$.
A chiral vertex operator $V^{\Delta(\theta_0)}_{\Delta(\theta_2),\Delta(\theta_1)}(z)$ is a map $V^{\Delta(\theta_0)}_{\Delta(\theta_2),\Delta(\theta_1)} : \mathcal{V}_{\theta_1}\rightarrow \mathcal{V}_{\theta_2}$. It is defined by the following relations:

\begin{subequations}\label{chiral}\begin{align}
\label{commutation}& [L_n,V^{\Delta(\theta_0)}_{\Delta(\theta_2),\Delta(\theta_1)}(z)]=z^n \lb z \partial_z+\Delta(\theta_0)(n+1)\rb V^{\Delta(\theta_0)}_{\Delta(\theta_2),\Delta(\theta_1)}(z),
	\\
\label{normalizationvertex} & V^{\Delta(\theta_0)}_{\Delta(\theta_2),\Delta(\theta_1)}(z)\left|\theta_1 \right>=N(\theta_2,\theta_0,\theta_1)~z^{\Delta(\theta_2)-\Delta(\theta_1)-\Delta(\theta_0)} \sum_{n\in \mathbb{Z}_{\geq 0}} v_n z^n,
\end{align}\end{subequations}
where $v_n \in \mathcal{V}_{\theta_2}$ and $v_0=\left| \theta_2\right>$, and where $N$ is a normalization factor which will be specified in Section \ref{section5}. 

\subsection{Four-point Virasoro conformal blocks} We now define a set of external conformal dimensions $\boldsymbol\theta$ by

\beq\label{theta}\boldsymbol\theta=\begin{pmatrix}\theta_0\\ \theta_t \\ \theta_1 \\ \theta_\infty \end{pmatrix}.\eeq 
The four-point Virasoro conformal block is defined by the expectation value of a composition of two chiral vertex operators as follows:
 \beq \label{CB}
 \mathcal F\lb b,\boldsymbol\theta,\sigma_s;z \rb = \frac{1}{N(\theta_\infty,\theta_1,\sigma_s)N(\sigma_s,\theta_t,\theta_0)}  \left\langle \theta_\infty\right|V_{\Delta(\theta_\infty),\Delta(\sigma_s)}^{\Delta(\theta_1)}\lb 1\rb V_{\Delta(\sigma_s),\Delta(\theta_0)}^{\Delta(\theta_t)}\lb z \rb\left|\theta_0\right\rangle.
 \eeq
The parameters $\boldsymbol\theta$ and $\sigma_s$ will be referred to as external and internal momenta. The block $\mathcal F$ admits a series expansion in $z$ which can be computed recursively using \eqref{chiral} and \eqref{virasoro}:
\beq\label{CBseries}
\mathcal F\lb b,\boldsymbol\theta,\sigma_s;z \rb = z^{\Delta(\sigma_s)-\Delta(\theta_0)-\Delta(\theta_t)}\lb1+\tfrac{(\Delta(\sigma_s)+\Delta(\theta_1)-\Delta(\theta_\infty))(\Delta(\sigma_s)+\Delta(\theta_t)-\Delta(\theta_0))}{2\Delta(\sigma_s)} z + \sum_{k=2}^\infty c_k z^k\rb.
\eeq
The discovery of the AGT relation \cite{AGT} between two-dimensional conformal field theories and four-dimensional supersymmetric gauge theories led to a closed-form expression for all the coefficients $c_k$. In our notation, the complete expansion of \eqref{CBseries} can be found in \cite[Eq. (3.1)]{LR}.

Several conjectures exist for the analytic properties of four-point Virasoro conformal blocks. The series in \eqref{CBseries} is believed to be convergent inside the unit disk $|z|<1$. Moreover, the only singularities of the conformal blocks as a function of $z$ are expected to be branch points at $0,1,\infty$ \cite{Zam87,HJP}. Under this assumption, conformal blocks are naturally defined for $z \in \mathbb{C} \setminus((-\infty,0 ] \cup [1,\infty))$. Finally, they are believed to be analytic in $\boldsymbol\theta$ and meromorphic in $\sigma_s$, with the only possible poles located at $\pm \sigma_s^{(m,n)}=-\frac{i}{2}(m b+n b^{-1})$ for $m,n \in \mathbb{Z}_{>0}$.
\subsection{Crossing transformations}

The linear span of four-point Virasoro conformal blocks forms an infinite-dimensional representation of $\Gamma(\Sigma_{0,4})=\text{PSL}_2(\mathbb{Z})$, the mapping class group of the four-puncture Riemann sphere. It is generated by the braiding $B$ and fusion $F$ moves, such that $F^2=(BF)^3=1$. The three ways of splitting four points on $\mathbb{C}\mathbb{P}^1$ into two pairs define the $s$-, $t$-, and $u$-channel bases for the space of conformal blocks. The cross-ratio argument of conformal blocks in these channels are chosen from $\{z,\frac{z}{z-1}\}$,$\{1-z,\frac{z-1}{z}\}$, and $\{\frac{1}{z},\frac{1}{1-z}\}$, respectively. 
The braiding move $B$ acts on the $s$-channel conformal blocks as follows:
$$
\mathcal F\lb b,\boldsymbol\theta,\sigma_s;z \rb = e^{\pm i\pi \lb \Delta(\sigma_s)-\Delta(\theta_0)-\Delta(\theta_t)\rb} \lb 1-z\rb^{-2\Delta(\theta_t)} \mathcal F\lb b,\tilde{\boldsymbol\theta},\sigma_s;\tfrac{z}{z-1} \rb,\quad \im z \gtrless 0,
$$
where $\tilde{\boldsymbol\theta}$ is obtained by performing the permutation $\theta_1 \leftrightarrow \theta_\infty$ on $\boldsymbol{\theta}$.
On the other hand, the fusion move was conjectured in \cite{pt1} and established in \cite{T03}. It is represented by the integral transform
 \beq \label{s-t}
 \qquad \mathcal F\lb b,\boldsymbol\theta,\sigma_s;z \rb = \int_{\mathbb R_+}d\sigma_t ~ 
  F\lb b,\boldsymbol\theta,\sigma_s,\sigma_t \rb
    \mathcal F\lb b,\hat{\boldsymbol\theta},\sigma_t;1-z \rb,
 \eeq
where $\hat{\boldsymbol\theta}$ is obtained by performing the permutation $\theta_0 \leftrightarrow \theta_1$ on $\boldsymbol{\theta}$. It will be convenient to explicitly write $\hat{\boldsymbol\theta}$  as
\beq\label{dualtheta}
\hat{\boldsymbol\theta} = K \boldsymbol \theta, \qquad K=\begin{pmatrix} 0 & 0 & 1 & 0 \\ 0 & 1 & 0 & 0 \\ 1 & 0 & 0 & 0 \\ 0 & 0 & 0 & 1 \end{pmatrix}.
\eeq
Finally, the kernel $F\lb b,\boldsymbol\theta,\sigma_s,\sigma_t \rb$ of the fusion move is the Virasoro fusion kernel. 

\section{The Virasoro fusion kernel}\label{section4}

In this section we recall the Ponsot-Teschner formula \cite{pt1} for the Virasoro fusion kernel. Let us also note that other representations for $F$ were found in \cite{TV14}. Moreover, we describe various symmetry and eigenfunction properties of $F$ which will be useful in its identification with the Ruijsenaars hypergeometric function \eqref{fonctionR}.

\subsection{Definition}

The Virasoro fusion kernel is defined by a contour integral involving two special functions $s_b(x)=g_b(x)/g_b(-x)$ and $g_b(x)$ which are closely related to the hyperbolic gamma function $G(a_-,a_+,x)=E(a_-,a_+,x)/E(a_-,a_+,-x)$ and $E(a_-,a_+,x)$, respectively. The exact relations are given in \eqref{sbgbGE}. The Ponsot-Teschner formula reads
\beq \label{fusion01}
\begin{split}
F\lb b,\boldsymbol\theta,\sigma_s,\sigma_t \rb = &  \prod_{\epsilon,\epsilon'=\pm} \frac{g_b \lb \epsilon \theta_1+\theta_{t}+\epsilon' \sigma_t\rb g_b \lb \epsilon \theta_0-\theta_\infty+\epsilon' \sigma_t \rb}{g_b \lb \epsilon \theta_0 + \theta_t + \epsilon' \sigma_s \rb g_b \lb \epsilon \theta_1-\theta_\infty+\epsilon' \sigma_s \rb} \prod_{\epsilon=\pm} \frac{g_b(\frac{iQ}2+2\epsilon \sigma_s)}{g_b(-\frac{iQ}2+2\epsilon \sigma_t)}
	\\
& \times \int_{\mathsf{F}} dx~\prod_{\epsilon=\pm1} \frac{s_b \lb x+ \epsilon \theta_1 \rb s_b \lb x+\epsilon\theta_0+\theta_\infty+\theta_t \rb}{s_b \lb x+\frac{iQ}{2}+\theta_\infty+\epsilon \sigma_s \rb s_b \lb x+\frac{i Q}{2}+\theta_t+\epsilon \sigma_t \rb}.
\end{split}
\eeq
When $b>0$, the integrand in (\ref{fusion01}) has eight vertical semi-infinite lines of poles, four of them increasing and the other four decreasing; the integration contour $\mathsf{F}$ runs from $-\infty$ to $+\infty$, separating the upper and lower sequences of poles. More generally, the fusion kernel \eqref{fusion01} can be extended to a meromorphic function of all of its parameters provided that $c \in \mathbb{C} \setminus \mathbb{R}_{\leq1}$, which corresponds to $b\notin i \mathbb{R}$. 

In the remainder of this article, we will make the following assumption on the parameters.
\begin{assumption}[Restriction on the parameters]\label{assumption}
We assume that
\beq 
0 < b < 1,\quad \boldsymbol\theta \in \mathbb{R}^4, \quad (\sigma_s,\sigma_t) \in \mathbb{R}^2.
\eeq
\end{assumption}

Assumption \ref{assumption} implies that the integration contour $\mathsf{F}$ is any curve going from $-\infty$ to $+\infty$ lying in the strip $\im x\in]-\tfrac{iQ}2,0[$. Moreover, Assumption \ref{assumption} is made primarily for simplicity; we expect all the results of this article to admit an analytic continuation to more general values of the parameters, such as $b\in \mathbb{C}\backslash i\mathbb{R},~ \boldsymbol\theta \in \mathbb{C}^4$ and $(\sigma_s,\sigma_t) \in \mathbb{C}^2$.

We are now going to derive various symmetry properties of the Virasoro fusion kernel.

\subsection{Symmetry properties} First, because the Virasoro conformal block \eqref{CBseries} is a function of the conformal dimensions $\Delta(x)=\frac{Q^2}{4}+x^2$, it follows that it is even in $\boldsymbol\theta$ and $\sigma_s$. Therefore the fusion transformation \eqref{s-t} implies that the Virasoro fusion kernel $F\lb b,\boldsymbol\theta,\sigma_s,\sigma_t \rb$ is even in $\boldsymbol\theta$ and $\sigma_s,\sigma_t$. Second, because $s_b(x)=s_{b^{-1}}(x)$ and $g_b(x)=g_{b^{-1}}(x)$, $F$ has the following symmetry:
\beq\label{modularF}
F\lb b,\boldsymbol\theta,\sigma_s,\sigma_t \rb = F\lb b^{-1},\boldsymbol\theta,\sigma_s,\sigma_t \rb.
\eeq
We now describe a duality transformation exchanging the internal momenta $\sigma_s$ and $\sigma_t$.
 \begin{proposition} \label{lemma1} 
The Virasoro fusion kernel given in \eqref{fusion01} satisfies 
\beq \label{identityF}
F\lb b,\boldsymbol\theta,\sigma_s,\sigma_t \rb = \frac{\alpha(\boldsymbol\theta,\sigma_s)\alpha(\boldsymbol\theta,-\sigma_s)}{\alpha(\hat{\boldsymbol\theta},\sigma_t)\alpha(\hat{\boldsymbol\theta},-\sigma_t)}
F\lb b,\hat{\boldsymbol\theta},\sigma_t,\sigma_s \rb ,
\eeq
where $\hat{\boldsymbol\theta}$ is related to $\boldsymbol\theta$ as in \eqref{dualtheta} and
\beq\label{alpha}
\alpha(\boldsymbol\theta,\sigma_s)=\frac{g_b\left(2\sigma_s+\frac{iQ}2 \right)g_b\left(2\sigma_s-\frac{iQ}2 \right)}{\prod _{\epsilon,\epsilon'=\pm} g_b(\sigma_s+\epsilon\theta_1 +\epsilon'\theta_\infty) g_b(\sigma_s+\epsilon\theta_0+\epsilon'\theta_t)}.
\eeq
\end{proposition}
\begin{proof}
The proof relies on the fact that the contour integral in \eqref{fusion01} is a hyperbolic Barnes integral $\mathcal{B}_h$ (see \eqref{hyperbolicbarnes} for its definition). In particular, $\mathcal{B}_h$ satisfies the remarkable identity \eqref{identitybarnes} which is crucial for the proof.

Following the definition \eqref{hyperbolicbarnes} and the relation $s_b(z)=G(b,b^{-1},z)$, it is easy to verify that the contour integral in \eqref{fusion01} takes the form 
\beq \label{4p6} \begin{split}
\int_{\mathsf{F}} dx~\prod_{\epsilon=\pm} \frac{s_b \lb x+ \epsilon \theta_1 \rb s_b \lb x+\epsilon\theta_0+\theta_\infty+\theta_t \rb}{s_b \lb x+\frac{iQ}{2}+\theta_\infty+\epsilon \sigma_s \rb s_b \lb x+\frac{i Q}{2}+\theta_t+\epsilon \sigma_t \rb} = \frac12 \mathcal{B}_h(b,b^{-1},\boldsymbol{u}), 
\end{split}
\eeq
where $\boldsymbol{u} \in \mathcal{G}_{iQ}$ and $\mathcal{G}_k$ is defined by \eqref{hyperplane}. Various choices for the eight parameters $\boldsymbol{u}$ lead to the same contour integral. Here we choose 
\beq\begin{split}
&u_1 = \tfrac{iQ}2+\theta_\infty+\sigma_s, \quad u_2 = \tfrac{iQ}2+\theta_\infty-\sigma_s, \quad u_3=\theta_1, \quad u_4=-\theta_1, \\
&u_5=-\theta_0-\theta_\infty-\theta_t, \quad u_6=\theta_0-\theta_\infty-\theta_t, \quad u_7 = \tfrac{iQ}2+\theta_t-\sigma_t, \quad u_8=\tfrac{iQ}2+\theta_t+\sigma_t.
\end{split}\eeq
We are now going to apply the identity \eqref{identitybarnes} satisfied by the hyperbolic Barnes integral. First, from \eqref{actionomega} the action of $\omega$ on $\boldsymbol u \in \mathcal{G}_{iQ}$ is given by
\beq\begin{split}
& (\omega \boldsymbol u)_1=\tfrac{iQ}2+\sigma_s, \quad (\omega \boldsymbol u)_2 = \tfrac{iQ}2-\sigma_s, \quad (\omega \boldsymbol u)_3 = \theta_1-\theta_\infty, \quad (\omega \boldsymbol u)_4=-\theta_1-\theta_\infty, \\
& (\omega \boldsymbol u)_5 = -\theta_0-\theta_t, \quad (\omega \boldsymbol u)_6=\theta_0-\theta_t, \quad (\omega \boldsymbol u)_7 = \tfrac{iQ}2+\theta_\infty+\theta_t-\sigma_t, \quad (\omega \boldsymbol u)_8=\tfrac{iQ}2+\theta_\infty+\theta_t+\sigma_t.
\end{split}\eeq
Second, a straightforward application of the identity \eqref{identitybarnes} yields
\beq\label{4p9}
\mathcal{B}_h(b,b^{-1},\boldsymbol{u}) = \mathcal{B}_h(b,b^{-1},\omega\boldsymbol{u}) \prod_{\epsilon,\epsilon'=\pm} s_b(\epsilon \sigma_s+\epsilon' \theta_1-\theta_\infty) s_b(\epsilon \sigma_t+\epsilon'\theta_0+\theta_\infty).
\eeq
Recalling \eqref{4p6} and using $s_b(z)=g_b(z)/g_b(-z)$, substitution of \eqref{4p9} into \eqref{fusion01} leads to
\beq\label{F01}
F\lb b,\boldsymbol\theta,\sigma_s,\sigma_t \rb = X_1~\mathcal{B}_h(b,b^{-1},\omega \boldsymbol{u}),
\eeq
where
\beq\label{X1}
X_1 = \frac12 \prod_{\epsilon,\epsilon'=\pm} \frac{g_b \lb \epsilon \theta_1+\theta_{t}+\epsilon' \sigma_t\rb g_b \lb \epsilon \theta_0+\theta_\infty+\epsilon' \sigma_t \rb}{g_b \lb \epsilon \theta_0 + \theta_t + \epsilon' \sigma_s \rb g_b \lb \epsilon \theta_1+\theta_\infty+\epsilon' \sigma_s \rb} \prod_{\epsilon=\pm} \frac{g_b(\frac{iQ}2+2\epsilon \sigma_s)}{g_b( -\frac{iQ}2+2\epsilon \sigma_t)}.
\eeq
On the other hand, using the representation \eqref{fusion01} and taking $x \to x+\theta_t$, it can be verified that we also have
\beq\label{F10}
F\lb b,I\hat{\boldsymbol\theta},\sigma_t,\sigma_s \rb = X_2~\mathcal{B}_h(b,b^{-1},\omega \boldsymbol{u}),
\eeq
where $\hat{\boldsymbol\theta}$ is given in \eqref{dualtheta}, $I=\text{diag}(1,-1,1,1)$, and
\beq\label{X2}
X_2 = \frac12 \prod_{\epsilon,\epsilon'=\pm} \frac{g_b \lb \epsilon \theta_0-\theta_{t}+\epsilon' \sigma_s\rb g_b \lb \epsilon \theta_1-\theta_\infty+\epsilon' \sigma_s \rb}{g_b \lb \epsilon \theta_1 - \theta_t + \epsilon' \sigma_t \rb g_b \lb \epsilon \theta_0-\theta_\infty+\epsilon' \sigma_t \rb} \prod_{\epsilon=\pm} \frac{g_b(\frac{iQ}2+2\epsilon \sigma_t)}{g_b( -\frac{iQ}2+2\epsilon \sigma_s)}.
\eeq
Since the Virasoro fusion kernel is even in the external momenta, we have $F\lb b,I\hat{\boldsymbol\theta},\sigma_t,\sigma_s \rb=F\lb b,\hat{\boldsymbol\theta},\sigma_t,\sigma_s \rb$. We deduce from \eqref{F01} and \eqref{F10} that the following identity holds:
\beq
F\lb b,\boldsymbol\theta,\sigma_s,\sigma_t \rb = \frac{X_1}{X_2}~F\lb b,I\hat{\boldsymbol\theta},\sigma_t,\sigma_s \rb,
\eeq
and a direct computation shows that
\beq
 \frac{X_1}{X_2} = \frac{\alpha(\boldsymbol\theta,\sigma_s)\alpha(\boldsymbol\theta,-\sigma_s)}{\alpha(\hat{\boldsymbol\theta},\sigma_t)\alpha(\hat{\boldsymbol\theta},-\sigma_t)},
\eeq
where $\alpha$ is given in \eqref{alpha}. Hence the identity \eqref{identityF} is obtained.
\end{proof}

\subsection{Difference equations} 
We now show that the Virasoro fusion kernel is a joint eigenfunction of four difference operators, the first two acting on the internal momentum $\sigma_s$ and the remaining two acting on its dual $\sigma_t$.

\subsubsection{First pair of difference equations}
Define the difference operator $H_F\lb b,\boldsymbol\theta,\sigma_s\rb$ by
\begin{equation} \label{HF}
 H_F\lb b,\boldsymbol\theta,\sigma_s\rb := H_F^+\lb b,\boldsymbol\theta,\sigma_s\rb e^{ib\partial_{\sigma_s}}+H_F^+\lb b,\boldsymbol\theta,-\sigma_s\rb e^{-ib\partial_{\sigma_s}}+H_F^0\lb b,\boldsymbol\theta,\sigma_s\rb,
\end{equation}
where
\beq \label{K}
 H_F^+\lb b,\boldsymbol\theta,\sigma_s\rb=4\pi^2 \frac{\Gamma \lb 1+2b^2-2ib\sigma_s \rb \Gamma \lb b^2-2i b \sigma_s \rb \Gamma \lb -2i b \sigma_s \rb \Gamma \lb 1+b^2-2ib \sigma_s \rb}{\prod_{\epsilon,\epsilon'=\pm1}\Gamma \lb \tfrac{bQ}2-ib(\sigma_s+\epsilon \theta_1+\epsilon' \theta_\infty)\rb \Gamma \lb \tfrac{bQ}2-ib (\sigma_s+\epsilon \theta_0+\epsilon' \theta_t )\rb},\eeq
and
\beq\label{H0} \begin{split}
&H_F^0\lb b,\boldsymbol\theta,\sigma_s\rb=-2\operatorname{cosh}{( 2\pi b (\theta_1+\theta_t+\tfrac{ib}2))}\\
&+4 \displaystyle \sum_{k=\pm} \frac{\prod_{\epsilon=\pm} \operatorname{cosh}{(\pi b(\epsilon \theta_\infty-\tfrac{ib}2-\theta_1-k\sigma_s))} \operatorname{cosh}{(\pi b(\epsilon \theta_0-\tfrac{ib}2-\theta_t-k\sigma_s))}}{\operatorname{sinh}{\lb 2\pi b(k\sigma_s+\frac{ib}2)\rb}\operatorname{sinh}{\lb 2\pi b k \sigma_s \rb}}.
\end{split} \eeq
In the CFT literature, the operator $H_F$ represents the action of a \textit{Verlinde loop operator} denoted $L_{\gamma_t}$ on the s-channel conformal blocks \eqref{CB}. Such an operator is obtained by computing a delicate quantum monodromy operation on the s-channel blocks \eqref{CB} along a certain closed curve $\gamma_t$ on the four-point Riemann sphere. The computation leading to $H_F$ was explicitly realized in \cite[Eq. (5.37)]{AGGTV}. It was also performed in \cite[Eq. (5.31)]{DGOT} in a different normalization for the conformal blocks.

We now present a direct proof of the fact that the Virasoro fusion kernel \eqref{fusion01} is an eigenfunction of $H_F$.

\begin{proposition} \label{prop1}
The Virasoro fusion kernel satisfies the following pair of difference equations:
\begin{subequations} \label{difference1} \begin{align} 
\label{differenceF1} & H_F\lb b,\boldsymbol\theta,\sigma_s\rb F\lb b,\boldsymbol\theta,\sigma_s,\sigma_t\rb = 2\operatorname{cosh}{\lb 2\pi b \sigma_t\rb} ~ F\lb b,\boldsymbol\theta,\sigma_s,\sigma_t\rb, \\
\label{differenceF2} & H_F\lb b^{-1},\boldsymbol\theta,\sigma_s\rb F\lb b,\boldsymbol\theta,\sigma_s,\sigma_t\rb = 2\operatorname{cosh}{\lb 2\pi b^{-1} \sigma_t\rb} ~ F\lb b,\boldsymbol\theta,\sigma_s,\sigma_t\rb.
\end{align}\end{subequations}
\end{proposition}
Although the pair of difference equations \eqref{difference1} is known, we are not aware of a direct proof of it. More precisely, equations \eqref{difference1} are special cases of the pentagon equation which is one relation of the Moore-Seiberg groupoid \cite{pt1,MS}. They also arise from the fact that the fusion transformation \eqref{s-t} diagonalizes the Verlinde loop operator $L_{\gamma_t}$ \cite{TV13}. We now present a direct proof of \eqref{difference1}\footnote{The proof is mostly inspired from the proof of \cite[Theorem 3.1]{R1999}.}.



\begin{proof}
We only need to prove \eqref{differenceF1} because of the symmetry \eqref{modularF} of $F$. It will be convenient to rewrite the definition \eqref{fusion01} of $F$ as follows:
\beq\label{fusionXYZ}
F\lb b,\boldsymbol\theta,\sigma_s,\sigma_t\rb = \int_{\mathsf{F}} dx ~ X_F(x,\sigma_s) Y_F(x,\sigma_t) Z_F(x),
\eeq
where the dependence of $X_F(x,\sigma_t), Y_F(x,\sigma_s)$ and $Z_F(x)$ on $b$ and $\boldsymbol\theta$ is omitted for simplicity, and
\beq\begin{split}
& X_F(x,\sigma_s)=\prod_{\epsilon=\pm} \frac{g_b \lb \tfrac{iQ}2+2\epsilon \sigma_s\rb}{s_b \lb x+\frac{iQ}{2}+\theta_\infty+\epsilon \sigma_s \rb} \prod_{\epsilon,\epsilon'=\pm} g_b \lb \epsilon \theta_0 + \theta_t + \epsilon' \sigma_s \rb^{-1} g_b \lb \epsilon \theta_1-\theta_\infty+\epsilon' \sigma_s \rb^{-1}, \\
&Y_F(x,\sigma_t)=\frac{\prod_{\epsilon,\epsilon'=\pm}g_b \lb \epsilon \theta_1+\theta_{t}+\epsilon' \sigma_t\rb g_b \lb \epsilon \theta_0-\theta_\infty+\epsilon' \sigma_t \rb}{\prod_{\epsilon=\pm}g_b(-\frac{iQ}2+2\epsilon \sigma_t)s_b \lb x+\frac{i Q}{2}+\theta_t+\epsilon \sigma_t \rb}, \\
& Z_F(x)=\prod_{\epsilon=\pm} s_b \lb x+ \epsilon \theta_1 \rb s_b \lb x+\epsilon\theta_0+\theta_\infty+\theta_t \rb.\end{split}\eeq
The proof relies on the following non-trivial identity for the action of the difference operator \eqref{HF} on the block $X_F$:
\beq \label{HFXF}
\frac{H_F\lb b,\boldsymbol\theta,\sigma_s\rb X_F(x,\sigma_s)}{X_F(x,\sigma_s)}=\psi(x,\sigma_s) + 2\cosh{(2 \pi b(x+\theta_t))},
\eeq
where
\beq\label{psi}
\psi(x,\sigma_s)=-4\prod_{\epsilon=\pm} \frac{\cosh \left(\pi  b \left(x+\frac{i b}{2}+\epsilon\theta_1\right)\right) \cosh \left(\pi  b \left(x+\frac{i b}{2}+\theta_\infty+\theta_t+\epsilon\theta_0 \right)\right)}{\sinh(\pi b(x+ib+\theta_\infty+\epsilon\sigma_s))}.
\eeq
Let us prove the identity \eqref{HFXF}. Using the identities \eqref{propgb} and \eqref{differencesb} satisfied by the functions $g_b$ and $s_b$, the left-hand side of \eqref{HFXF} takes the explicit form 
\beq\label{HFXF1}
\frac{H_F\lb b,\boldsymbol\theta,\sigma_s\rb X_F(x,\sigma_s)}{X_F(x,\sigma_s)} = H_F^0\lb b,\boldsymbol\theta,\sigma_s\rb + \phi(x,\sigma_s) + \phi(x,-\sigma_s),
\eeq
where $H_F^0$ is given in \eqref{H0} and
\beq\begin{split}
&\phi(x,\sigma_s)\\
&=-\frac{4 \sinh (\pi  b (x+\theta_\infty-\sigma_s))\prod_{\epsilon=\pm} \cosh (\pi  b (\epsilon\theta_0+\tfrac{i b}{2}-\theta_t+\sigma_s)) \cosh(\pi  b(\epsilon\theta_1+\tfrac{i b}{2}+\theta_\infty+\sigma_s))}{\sinh (2\pi  b (\sigma_s+\tfrac{i b}2)) \sinh (2 \pi  b \sigma_s) \sinh (\pi  b (x+i b+\theta_\infty+\sigma_s))} .
\end{split} \eeq
Hence the proof of \eqref{HFXF} is equivalent to proving the identity $f_1(\sigma_s)=f_2(\sigma_s)$ where
\beq\label{4p18}\begin{split}
& f_1(\sigma_s) = \psi(x,\sigma_s) + 2 \cosh{(2 \pi b(x+\theta_t))}, \\
& f_2(\sigma_s) = H_F^0\lb b,\boldsymbol\theta,\sigma_s\rb + \phi(x,\sigma_s) + \phi(x,-\sigma_s).
\end{split}\eeq
We proceed as follows: both $f_1(\sigma_s)$ and $f_2(\sigma_s)$ are meromorphic functions of $\sigma_s$. It can also be easily be checked that they are even and $i b^{-1}$-periodic in $\sigma_s$. Moreover, we have the following asymptotics:
\beq\label{asymps}\begin{split}
&\lim_{Re(\sigma_s) \to \pm\infty} H_F^0\lb b,\boldsymbol\theta,\sigma_s\rb = 0 = \lim_{Re(\sigma_s) \to \pm\infty} \psi(x,\sigma_s), \\
     & \lim_{Re(\sigma_s) \to \pm \infty} \phi(x,\sigma_s) = e^{\mp2\pi b(x+\theta_t)},
     \end{split}\eeq
which entail the asymptotics
\beq\label{4p23}
\lim_{Re(\sigma_s) \to \pm\infty} f_1(\sigma_s) = 2\cosh{(2\pi b(x+\theta_t))} = \lim_{Re(\sigma_s) \to \pm\infty} f_2(\sigma_s).
\eeq
We now show that $f_1(\sigma_s)$ and $f_2(\sigma_s)$  have equal $\sigma_s$-residues in a horizontal period strip $\im \sigma_s \in [0,ib^{-1}]$. It can first be verified that $f_2(\sigma_s)$ has vanishing residues at $\sigma_s=\pm\tfrac{ib}2,0$. However it has non-zero residues at $\sigma_s=\pm(ib+\theta_\infty+x)$ which are equal to
\beq\label{residue}\begin{split}
&\text{Res}(f_2(\sigma_s))|_{\sigma_s=\pm(ib+\theta_\infty+x)}\\
&=\pm\frac{4}{\pi b} \sinh{(2 \pi  b (i b+\theta_\infty+x))}^{-1} \prod _{\epsilon =\pm} \cosh(\pi  b (\tfrac{i b}{2}+x+\epsilon\theta_1)) \cosh(\pi  b(\tfrac{i b}{2}+\theta_\infty+\theta_t+x+\epsilon\theta_0)).\end{split}\eeq
On the other hand, the only residues of $f_1(\sigma_s)$ are located at $\sigma_s=\pm(ib+\theta_\infty+x)$ and straightforward computations show that they are also equal to the right-hand side of \eqref{residue}. Hence we have shown that the function $f_1(\sigma_s)-f_2(\sigma_s)$ is bounded and holomorphic in a horizontal period strip $\im \sigma_s \in [0,ib^{-1}]$. These properties extend to the whole complex plane by $ib^{-1}$-periodicity in $\sigma_s$. Therefore Liouville theorem ensures that the function $f_1(\sigma_s)-f_2(\sigma_s)$ is constant everywhere. By evaluation at $\text{Re}(\sigma_s)=+\infty$ using \eqref{4p23}, we deduce that $f_1(\sigma_s)=f_2(\sigma_s)$. Hence the identity \eqref{HFXF} is proved.

We now let the difference operator $H_F$ defined in \eqref{HF} act on the Virasoro fusion kernel \eqref{fusionXYZ}. Using \eqref{HFXF}, we have
\beq \label{step1} \begin{split}
& H_F\lb b,\boldsymbol\theta,\sigma_s\rb F\lb b,\boldsymbol\theta,\sigma_s,\sigma_t\rb = \int_{\mathsf{F}} dx~X_F(x,\sigma_s) Y_F(x,\sigma_t) Z_F(x) \psi(x,\sigma_s) \\
     &+2 \int_{\mathsf{F}} dx~X_F(x,\sigma_s) Y_F(x,\sigma_t) Z_F(x)\cosh{(2\pi b(x+\theta_t))}.
\end{split} \eeq
The next step of the proof utilizes the fact that the blocks $X_F$, $Y_F$ and $Z_F$ satisfy the following identities:
\begin{subequations} \label{eqXYZ} \begin{align}
\label{block1} & \frac{X_F(x-ib,\sigma_s)}{X_F(x,\sigma_s)} = 2 \cosh (2 \pi  b \sigma_s)-2 \cosh (2\pi  b (x+\theta_\infty)), \\
\label{block2} & \frac{Y_F(x-ib,\sigma_t)}{Y_F(x,\sigma_t)} = 2 \cosh (2 \pi  b \sigma_t)-2 \cosh{(2\pi b(x+\theta_t))}, \\
\label{block3} & \frac{Z_F(x-ib)}{Z_F(x)} = \frac{1}{16 \prod_{\epsilon=\pm} \cosh(\pi b(x-\tfrac{i b}{2}+\epsilon\theta_1)) \cosh(\pi  b(x-\tfrac{i b}{2}+\theta_\infty+\theta_t+\epsilon\theta_0))}.
\end{align} \end{subequations}
We now perform a contour shift together with $x\to x-ib$ in the first line of \eqref{step1} and we use \eqref{block2}. Assumption \ref{assumption} ensures that the contour does not cross any pole. We obtain
\beq\label{4p27}\begin{split}
H_F\lb b,\boldsymbol\theta,\sigma_s\rb & F\lb b,\boldsymbol\theta,\sigma_s,\sigma_t\rb = 2\cosh{(2\pi b \sigma_t)} \int_{\mathsf{F}} dx~X_F(x-ib,\sigma_s) Y_F(x,\sigma_t) Z_F(x-ib) \psi(x-ib,\sigma_s) \\
     & - 2 \int_{\mathsf{F}} dx~X_F(x-ib,\sigma_s) Y_F(x,\sigma_t) Z_F(x-ib) \psi(x-ib,\sigma_s)\cosh{(2 \pi b(x+\theta_t))} \\
     & + 2 \int_{\mathsf{F}} dx~X_F(x,\sigma_s) Y_F(x,\sigma_t) Z_F(x) \cosh{(2 \pi b(x+\theta_t))}.
\end{split}\eeq
Finally, the identity
\beq \label{4p28}
\psi(x-ib,\sigma_s) X_F(x-ib,\sigma_s)Z_F(x-ib)=X_F(x,\sigma_s)Z_F(x)
\eeq
implies that the last two lines in \eqref{4p27} cancel, and that the first line in $\eqref{4p27}$ yields the desired result.
\end{proof}
\subsubsection{Second pair of difference equations} 
The duality transformation \eqref{identityF} exchanging the internal momenta $\sigma_s$ and $\sigma_t$ implies that the Virasoro fusion kernel \eqref{fusion01} satisfies another pair of difference equations where the difference operators act on $\sigma_t$. We introduce the dual operator $\tilde{H}_F$ by
 \beq \label{tildeHF} \tilde H_F\lb b,\hat{\boldsymbol\theta},\sigma_t\rb := \tilde{H}_F^+\lb b,\hat{\boldsymbol\theta},\sigma_t\rb e^{ib\partial_{\sigma_t}}+\tilde{H}_F^+\lb b,\hat{\boldsymbol\theta},-\sigma_t\rb e^{-ib\partial_{\sigma_t}}+H_F^0\lb b,\hat{\boldsymbol\theta},\sigma_t\rb, \eeq
where $H_F^0$ is given in \eqref{H0}, $\hat{\boldsymbol\theta}$ is defined by \eqref{dualtheta} and
     \beq \label{tildeK}\begin{split}
  \tilde{H}_F^+\lb b,\hat{\boldsymbol\theta},\sigma_t\rb = \frac{4\pi^2~\Gamma \left(1-b^2+2 i b\sigma_t \right) \Gamma (1+2ib\sigma_t) \Gamma \left(2 i b \sigma_t-2 b^2\right) \Gamma \left(2 i b \sigma_t-b^2\right)}{\prod _{\epsilon,\epsilon'=\pm} \Gamma \left(\frac{1-b^2}2+ib \left(\sigma_t+\epsilon \theta_0+\epsilon' \theta_\infty\right)\right) \Gamma \left(\frac{1-b^2}2+ib \left(\sigma_t+\epsilon \theta_1+\epsilon' \theta_t\right)\right)}.
     \end{split}\eeq
We next show that the Virasoro fusion kernel is also an eigenfunction of this difference operator.
\begin{proposition}\label{prop4p2}
The Virasoro fusion kernel satisfies the dual pair of difference equations
\begin{subequations}\label{difference2} \begin{align}
\label{differenceF3} & \tilde{H}_F\lb b,\hat{\boldsymbol\theta},\sigma_t\rb F\lb b,\boldsymbol\theta,\sigma_s,\sigma_t\rb = 2\operatorname{cosh}{\lb 2\pi b \sigma_s\rb} ~ F\lb b,\boldsymbol\theta,\sigma_s,\sigma_t\rb, \\
\label{differenceF4} & \tilde{H}_F\lb b^{-1},\hat{\boldsymbol\theta},\sigma_t\rb F\lb b,\boldsymbol\theta,\sigma_s,\sigma_t\rb = 2\operatorname{cosh}{\lb 2\pi b^{-1} \sigma_s\rb} ~ F\lb b,\boldsymbol\theta,\sigma_s,\sigma_t\rb.
\end{align}\end{subequations}
\end{proposition}
\begin{proof}
Again, we only need to prove \eqref{differenceF3} thanks to the symmetry \eqref{modularF} of $F$. Let us first apply the duality transformation $\boldsymbol\theta \to \hat{\boldsymbol\theta}$ and $\sigma_s \leftrightarrow \sigma_t$ to the equation \eqref{differenceF1}:
\beq
\label{difference1prime} H_F\lb b,\hat{\boldsymbol\theta},\sigma_t\rb F\lb b,\hat{\boldsymbol\theta},\sigma_t,\sigma_s\rb = 2\operatorname{cosh}{\lb 2\pi b \sigma_s\rb} ~ F\lb b,\hat{\boldsymbol\theta},\sigma_t,\sigma_s\rb,\eeq
where $H_F$ is given in \eqref{HF}. Substituting the identity \eqref{identityF} into \eqref{difference1prime}, we obtain 
\beq\label{4p33}
\lb\alpha(\hat{\boldsymbol\theta},\pm \sigma_t)^{-1}\circ H_F\lb b,\hat{\boldsymbol\theta},\sigma_t\rb \circ\alpha(\hat{\boldsymbol\theta},\pm \sigma_t)\rb F\lb b,\boldsymbol\theta,\sigma_s,\sigma_t\rb = 2\operatorname{cosh}{\lb 2\pi b \sigma_s\rb} ~ F\lb b,\boldsymbol\theta,\sigma_s,\sigma_t\rb,\eeq
where for the sake of brevity we denoted $f(a,\pm b)=f(a,b)f(a,-b)$. The operator on the left-hand side of \eqref{4p33} takes the form
\beq\begin{split}
\alpha(\hat{\boldsymbol\theta},\pm \sigma_t)^{-1} & \circ H_F\lb b,\hat{\boldsymbol\theta},\sigma_t\rb\circ \alpha(\hat{\boldsymbol\theta},\pm \sigma_t) = H_F^0\lb b,\hat{\boldsymbol\theta},\sigma_t\rb\\
&+\frac{\alpha(\hat{\boldsymbol\theta},\sigma_t+ib)}{\alpha(\hat{\boldsymbol\theta},\sigma_t)} \frac{\alpha(\hat{\boldsymbol\theta},-\sigma_t-ib)}{\alpha(\hat{\boldsymbol\theta},-\sigma_t)} H_F^+\lb b,\hat{\boldsymbol\theta},\sigma_t\rb e^{ib\partial_{\sigma_t}} \\
&+ \frac{\alpha(\hat{\boldsymbol\theta},\sigma_t-ib)}{\alpha(\hat{\boldsymbol\theta},\sigma_t)} \frac{\alpha(\hat{\boldsymbol\theta},-\sigma_t+ib)}{\alpha(\hat{\boldsymbol\theta},-\sigma_t)}  H_F^+\lb b,\hat{\boldsymbol\theta},-\sigma_t\rb e^{-ib\partial_{\sigma_t}}.
\end{split}\eeq
A tedious but straightforward computation using the properties \eqref{propgb} and \eqref{differencesb} of the functions $g_b$ and $s_b$ shows that for $k=\pm1$ we have
\beq
\frac{\alpha(\hat{\boldsymbol\theta},\sigma_t +k ib)}{\alpha(\hat{\boldsymbol\theta},\sigma_t)} \frac{\alpha(\hat{\boldsymbol\theta},-\sigma_t-k ib)}{\alpha(\hat{\boldsymbol\theta},-\sigma_t)} H_F^+\lb b,\hat{\boldsymbol\theta},k\sigma_t\rb =  \tilde{H}_F^+\lb b,\hat{\boldsymbol\theta},k\sigma_t\rb,
\eeq
which finally implies the identity
\beq
\alpha(\hat{\boldsymbol\theta},\pm \sigma_t)^{-1}\circ H_F\lb b,\hat{\boldsymbol\theta},\sigma_t\rb\circ \alpha(\hat{\boldsymbol\theta},\pm \sigma_t) = \tilde H_F\lb b,\hat{\boldsymbol\theta},\sigma_t\rb.
\eeq
In view of \eqref{4p33}, this completes the proof.
 \end{proof}
 
In summary, we have shown that the Virasoro fusion kernel satisfies two pairs of difference equations given in \eqref{difference1} and \eqref{difference2}. 

 
  
\section{Main result} \label{section5}

The renormalized Ruijsenaars hypergeometric function \eqref{fonctionR} and the Virasoro fusion kernel \eqref{fusion01} were studied in Sections \ref{section2} and \ref{section4}, respectively. Although they appear in different contexts, the two functions resemble each other. Both of them are proportional to a hyperbolic Barnes integral \eqref{hyperbolicbarnes} and are joint eigenfunctions of four difference operators. The main result of this section is Theorem \ref{theoremRF}. It shows that $F$ and $R_\text{ren}$ are the same function up to normalization.

There are at least three obstacles in the identification of $F$ and $R_\text{ren}$. First, the two functions are expressed in terms of different parameters. Second, for instance the definition of $F$ depend on the normalization of the s- and t-channel conformal blocks. Finally, both $F$ and $R_\text{ren}$ possess various equivalent representations because the hyperbolic Barnes integral satisfies the nontrivial identities \eqref{identitybarnes}. Hence it appears hopeless to find an exact relation between $F$ and $R_\text{ren}$ only by comparing their integral representation.


Instead, the key to the identification of $F$ and $R_\text{ren}$ will be to compare and match the four difference equations that they satisfy. Let us first compare the difference operators $H_F$ defined by \eqref{HF} and $H_{RvD}$ given in \eqref{hamiltonienR}. $H_F$ and $H_{RvD}$ appear in the left-hand sides of the difference equations \eqref{differenceF1} and \eqref{differenceR1} satisfied by $F$ and $R_\text{ren}$, respectively. Of particular importance are the "potentials" $H_F^0$ defined by \eqref{H0} and $V$ given in \eqref{potentialV}, because they remain invariant under any change of normalization. The first step in the identification of $F$ and $R_\text{ren}$ is to show that $H_F^0$ and $V$ are equal under a certain parameters correspondence.


\subsection{Ruijsenaars/CFT parameters correspondence}\label{section5p1}
We now provide a parameters correspondence between the Virasoro fusion kernel $F\lb b,\boldsymbol\theta,\sigma_s,\sigma_t\rb$ and the Ruijsenaars hypergeometric function $R_\text{ren}(a_-,a_+,\boldsymbol\gamma,v,\hat{v})$. First of all, we have
\beq\label{5p2}
b=a_-, \qquad b^{-1}=a_+.
\eeq
The second relation between the pairs $(v,\hat v)$, and $(\sigma_s,\sigma_t)$ is simply given by
\beq
v=\sigma_s, \qquad \hat v=\sigma_t.
\eeq
Moreover, the couplings $\boldsymbol{\gamma}$ defined by \eqref{gamma} and the external momenta $\boldsymbol\theta$ given in \eqref{theta} are related by
\beq\label{parametersR}
\boldsymbol\gamma=L \boldsymbol\theta,\eeq
where the matrix $L$ reads
\beq\label{L}
L = \begin{pmatrix}-i & -i & 0 & 0 \\ 0 & 0 & -i & i \\ i & -i & 0 & 0 \\ 0 & 0 & -i & -i \end{pmatrix}.
\eeq
Finally, the dual sets of parameters $\hat{\boldsymbol{\gamma}}$ and $\hat{\boldsymbol{\theta}}$ respectively given in \eqref{gamma} and \eqref{dualtheta} can be related as follows. It is straightforward to verify that the matrices $J$ defined by \eqref{J}, K given in \eqref{dualtheta} and $L$ satisfy
\beq
JL=LK,
\eeq
which implies
\beq\label{dualparametersR}
\hat{\boldsymbol\gamma}=L \hat{\boldsymbol\theta}.\eeq
We emphasize that the parameters correspondence described above can be rescaled thanks to the property \eqref{scaleR} of the function $R_\text{ren}$.

It can now be verified that the potentials $H_F^0$ defined by \eqref{H0} and $V$ given in \eqref{potentialV} are equal:
\beq\label{VH0}
V(b,b^{-1},L\boldsymbol\theta,\sigma_s) = H_F^0\lb b,\boldsymbol\theta,\sigma_s\rb.
\eeq
This identification suggests that the difference operator $H_{RvD}$ defined in \eqref{hamiltonienR} can be obtained from $H_F$ given in \eqref{HF} after a proper change of normalization. It also suggests that the difference equation \eqref{differenceF1} satisfied by $F$ can be mapped to the difference equation \eqref{differenceR1} satisfied by $R_\text{ren}$.

\subsection{Renormalized fusion kernel $F_{\text{ren}}$}

\subsubsection{Some requirements}
The second step in the identification of $F$ and $R_\text{ren}$ is to find a renormalized version $F_{\text{ren}}$ of $F$ which satisfies a difference equation of the form \eqref{differenceR1}, but also three other equations of the form \eqref{differenceR2}, \eqref{differenceR3}, \eqref{differenceR4}. This requirement can be reformulated as follows. The difference equation \eqref{differenceR1} satisfied by the function $R_\text{ren}$ implies the three remaining equations \eqref{differenceR2}, \eqref{differenceR3} and \eqref{differenceR4} because $R_\text{ren}$ satisfies the identities \eqref{modulardual} and \eqref{selfdual}. Moreover, using the parameters correspondence of Section \ref{section5p1}, the self-duality \eqref{selfdual} of $R_\text{ren}$ becomes
\beq
R_{\text{ren}}(b,b^{-1},L\boldsymbol{\theta},\sigma_s,\sigma_t)=R_{\text{ren}}(b,b^{-1},L\hat{\boldsymbol{\theta}},\sigma_t,\sigma_s).
\eeq
This equation can be compared to the property \eqref{identityF} of the Virasoro fusion kernel: the function $R_\text{ren}$ is self-dual, while $F$ is not. We deduce that the key to the identification of $F$ and $R_\text{ren}$ is to find a renormalized version $F_\text{ren}$ of $F$ which (i) satisfies a difference equation of the type \eqref{differenceR1}, (ii) is invariant under $b\to b^{-1}$, (iii) is self-dual. 
\subsubsection{Definition of $F_{\text{ren}}$}
We now present a renormalized version $F_{\text{ren}}$ of $F$ which satisfies the three properties (i)-(iii). First, we define renormalized conformal blocks $\mathcal{F}_\text{ren}$ by
\beq
\mathcal F_\text{ren}\lb b,\boldsymbol\theta,\sigma_s;z\rb =  N(\theta_\infty,\theta_1,\sigma_s) N(\sigma_s,\theta_t,\theta_0) \mathcal F\lb b,\boldsymbol\theta,\sigma_s;z\rb,
\eeq 
where $\mathcal F$ is given in \eqref{CB}. Moreover, we choose the following normalization factor:
\beq \label{normalizationR}
N(\theta_3,\theta_2,\theta_1)=\frac{g_b(-\theta_1-\theta_2-\theta_3) g_b(\theta_1-\theta_2-\theta_3) g_b(-\theta_1-\theta_2+\theta_3) g_b(\theta_1-\theta_2+\theta_3)}{g_b(-2 \theta_1+\frac{i Q}{2}) g_b(-2 \theta_2+\frac{i Q}{2}) g_b(2 \theta_3+\frac{i Q}{2})}.
\eeq
We are now going to rewrite the fusion transformation \eqref{s-t} in terms of $\mathcal{F}_\text{ren}$. We first introduce a weight function\footnote{The weight function $\omega$ was introduced by Ruijsenaars in \cite[Eq. (1.15)]{R2003} in order to construct a unitary eigenfunction transform associated to the function $R_\text{ren}$.}
\beq\label{defomega}
\omega\lb b,\boldsymbol\theta,\sigma_s\rb=\mu\lb b,\boldsymbol\theta,\sigma_s\rb \mu\lb b,\boldsymbol\theta,-\sigma_s\rb,
\eeq
where
\beq \label{defmu} \mu\lb b,\boldsymbol\theta,\sigma_s\rb=\frac{s_b(\tfrac{iQ}2+2\sigma_s)}{\prod_{\epsilon=\pm1}s_b \lb \sigma_s-\theta_1+\epsilon \theta_\infty \rb s_b \lb \sigma_s-\theta_t+\epsilon\theta_0\rb}.
\eeq 
The fusion transformation \eqref{s-t} can now be rewritten as
\beq \label{s-t-tilde}
 \qquad \mathcal F_\text{ren} \lb b,\boldsymbol\theta,\sigma_s;z\rb=\int_{\mathbb R_+} d\sigma_t ~ \omega\lb b,\hat{\boldsymbol\theta},\sigma_t\rb ~ F_{\text{ren}}\lb b,\boldsymbol\theta,\sigma_s,\sigma_t\rb \mathcal F_\text{ren}\lb b,\hat{\boldsymbol\theta},\sigma_t;1-z\rb,
 \eeq
 where the renormalized Virasoro fusion kernel $F_{\text{ren}}$ is defined by
\beq \label{tildeF}
F_{\text{ren}}\lb b,\boldsymbol\theta,\sigma_s,\sigma_t\rb= \omega\lb b,\hat{\boldsymbol\theta},\sigma_t\rb^{-1} \frac{N(\theta_\infty,\theta_1,\sigma_s) N(\sigma_s,\theta_t,\theta_0)}{N(\theta_\infty,\theta_0,\sigma_t) N(\sigma_t,\theta_t,\theta_1)} ~ F\lb b,\boldsymbol\theta,\sigma_s,\sigma_t\rb,
 \eeq
 and where $F$ is defined in \eqref{fusion01}. In the next three propositions, we show that $F_{\text{ren}}$ satisfies the desired three properties (i)-(iii).
\begin{proposition}
The following identity holds:
 \beq\label{modulartildeF}
F_{\text{ren}}\lb b^{-1},\boldsymbol\theta,\sigma_s,\sigma_t\rb = F_{\text{ren}}\lb b,\boldsymbol\theta,\sigma_s,\sigma_t\rb.
\eeq
 \end{proposition}
\begin{proof}
Equation \eqref{modulartildeF} follows from the symmetries $s_b(x)=s_{b^{-1}}(x)$ and $g_b(x)=g_{b^{-1}}(x)$.
\end{proof}
We now prove that $F_{\text{ren}}$ is self-dual.
\begin{proposition}
 The renormalized Virasoro fusion kernel $F_{\text{ren}}$ satisfies 
 \beq\label{identitytildeF}
F_{\text{ren}}\lb b,\boldsymbol\theta,\sigma_s,\sigma_t\rb = F_{\text{ren}}\lb b,\hat{\boldsymbol\theta},\sigma_t,\sigma_s\rb.
\eeq
\end{proposition}
\begin{proof}
The proof consists of writing an explicit relation between $F_{\text{ren}}\lb b,\boldsymbol\theta,\sigma_s,\sigma_t\rb$ and $F_{\text{ren}}\lb b,\hat{\boldsymbol\theta},\sigma_t,\sigma_s\rb$ using the definition \eqref{tildeF}  and the identity \eqref{identityF}. First, from \eqref{tildeF} we have
\beq\label{Frenhat}
F_{\text{ren}}\lb b,\hat{\boldsymbol\theta},\sigma_t,\sigma_s\rb= \omega\lb b,\boldsymbol\theta,\sigma_s\rb^{-1} \frac{N(\theta_\infty,\theta_0,\sigma_t) N(\sigma_t,\theta_t,\theta_1)}{N(\theta_\infty,\theta_1,\sigma_s) N(\sigma_s,\theta_t,\theta_0)} ~F\lb b,\hat{\boldsymbol\theta},\sigma_t,\sigma_s\rb.
\eeq
Substituting the identity \eqref{identityF} into \eqref{Frenhat}, we have
\beq
F_{\text{ren}}\lb b,\hat{\boldsymbol\theta},\sigma_t,\sigma_s\rb = \omega\lb b,\boldsymbol\theta,\sigma_s\rb^{-1} \frac{N(\theta_\infty,\theta_0,\sigma_t) N(\sigma_t,\theta_t,\theta_1)}{N(\theta_\infty,\theta_1,\sigma_s) N(\sigma_s,\theta_t,\theta_0)} ~ \frac{\alpha(\hat{\boldsymbol\theta},\sigma_t)\alpha(\hat{\boldsymbol\theta},-\sigma_t)}{\alpha(\boldsymbol\theta,\sigma_s)\alpha(\boldsymbol\theta,-\sigma_s)} F\lb b,\boldsymbol\theta,\sigma_s,\sigma_t\rb,
\eeq
where $\alpha$ is given in \eqref{alpha}. Using once again the definition \eqref{tildeF}, we obtain
\beq\label{5p20}\begin{split}
& F_{\text{ren}}\lb b,\hat{\boldsymbol\theta},\sigma_t,\sigma_s\rb = \omega\lb b,\boldsymbol\theta,\sigma_s\rb^{-1} \frac{N(\theta_\infty,\theta_0,\sigma_t) N(\sigma_t,\theta_t,\theta_1)}{N(\theta_\infty,\theta_1,\sigma_s) N(\sigma_s,\theta_t,\theta_0)} \\
& \times \frac{\alpha(\hat{\boldsymbol\theta},\sigma_t)\alpha(\hat{\boldsymbol\theta},-\sigma_t)}{\alpha(\boldsymbol\theta,\sigma_s)\alpha(\boldsymbol\theta,-\sigma_s)} ~ \omega\lb b,\hat{\boldsymbol\theta},\sigma_t\rb \frac{N(\theta_\infty,\theta_0,\sigma_t) N(\sigma_t,\theta_t,\theta_1)} {N(\theta_\infty,\theta_1,\sigma_s) N(\sigma_s,\theta_t,\theta_0)} ~ F_{\text{ren}}\lb b,\boldsymbol\theta,\sigma_s,\sigma_t\rb .
\end{split} \eeq
Finally, a direct computation using $s_b(x)=g_b(x)/g_b(-x)$ shows that all the prefactors in the right-hand side of \eqref{5p20} cancel out. Therefore the identity \eqref{identitytildeF} is obtained. \end{proof}

\subsection{Difference equations for $F_{\text{ren}}$}
We now show that the renormalized Virasoro fusion kernel and the renormalized Ruijsenaars hypergeometric function satisfy the same four difference equations. 
\begin{proposition}\label{prop2}
The renormalized fusion kernel $F_{\text{ren}}$ defined by \eqref{tildeF} satisfies the following difference equations:
\begin{subequations} \label{differencetildeF} \begin{align}
\label{differencetildeF1} &  H_{RvD}(b,b^{-1},L\boldsymbol\theta,\sigma_s) F_{\text{ren}}\lb b,\boldsymbol\theta,\sigma_s,\sigma_t\rb  = 2\operatorname{cosh}{(2\pi b \sigma_t)}~F_{\text{ren}}\lb b,\boldsymbol\theta,\sigma_s,\sigma_t\rb , \\
\label{differencetildeF2} & H_{RvD}(b^{-1},b,L\boldsymbol\theta,\sigma_s) F_{\text{ren}}\lb b,\boldsymbol\theta,\sigma_s,\sigma_t\rb = 2\operatorname{cosh}{(2\pi b^{-1} \sigma_t)}~F_{\text{ren}}\lb b,\boldsymbol\theta,\sigma_s,\sigma_t\rb , \\
\label{differencetildeF3} & H_{RvD}(b,b^{-1},L\hat{\boldsymbol\theta},\sigma_t) F_{\text{ren}}\lb b,\boldsymbol\theta,\sigma_s,\sigma_t\rb = 2\operatorname{cosh}{(2\pi b \sigma_s)}~F_{\text{ren}}\lb b,\boldsymbol\theta,\sigma_s,\sigma_t\rb , \\
\label{differencetildeF4} & H_{RvD}(b^{-1},b,L\hat{\boldsymbol\theta},\sigma_t) F_{\text{ren}}\lb b,\boldsymbol\theta,\sigma_s,\sigma_t\rb = 2\operatorname{cosh}{(2\pi b^{-1} \sigma_s)}~F_{\text{ren}}\lb b,\boldsymbol\theta,\sigma_s,\sigma_t\rb ,
\end{align}\end{subequations}
where the matrix $L$ is given in \eqref{L} and $H_{RvD}$ is defined by \eqref{hamiltonienR}.
\end{proposition}

\begin{proof}
We only need to prove \eqref{differencetildeF1} because of the symmetries \eqref{modulartildeF} and \eqref{identitytildeF} of $F_\text{ren}$. Moreover, it will be convenient to rewrite \eqref{tildeF} as $F=(P(\sigma_t)/Q(\sigma_s))F_{\text{ren}}$, where
\beq\begin{split}
&P(\sigma_t) = \omega\lb b,\hat{\boldsymbol\theta},\sigma_t\rb N(\theta_\infty,\theta_0,\sigma_t) N(\sigma_t,\theta_t,\theta_1), \\
&Q(\sigma_s) = N(\theta_\infty,\theta_1,\sigma_s) N(\sigma_s,\theta_t,\theta_0).
\end{split} \eeq
Substituting $F=(P(\sigma_t)/Q(\sigma_s))F_{\text{ren}}$ into the difference equation \eqref{differenceF1} satisfied by $F$, we obtain
 \beq\label{5p15}
\lb \frac{Q(\sigma_s)}{P(\sigma_t)}\circ H_F\lb b,\boldsymbol\theta,\sigma_s\rb\circ \frac{P(\sigma_t)}{Q(\sigma_s)} \rb F_{\text{ren}}\lb b,\boldsymbol\theta,\sigma_s,\sigma_t\rb = 2\operatorname{cosh}{\lb 2\pi b \sigma_t\rb} ~ F_{\text{ren}}\lb b,\boldsymbol\theta,\sigma_s,\sigma_t\rb.
\eeq
The factors $P(\sigma_t)$ in the left-hand side of \eqref{5p15} cancel out, since the difference operator $H_F\lb b,\boldsymbol\theta,\sigma_s\rb$ acts on the variable $\sigma_s$. Therefore, the proof of \eqref{differencetildeF1} consists of verifying the following identity:
\beq \label{htildehR1}
Q(\sigma_s)\circ H_F\lb b,\boldsymbol\theta,\sigma_s\rb \circ Q(\sigma_s)^{-1} = H_{RvD}(b,b^{-1},L\boldsymbol\theta,\sigma_s),
\eeq
where $H_F$ and $H_{RvD}$ are respectively defined by \eqref{HF} and \eqref{hamiltonienR}. We have
\beq\begin{split}\label{5p18}
& Q(\sigma_s)\circ H_F\lb b,\boldsymbol\theta,\sigma_s\rb\circ Q(\sigma_s)^{-1}  \\
     & = H_F^0\lb b,\boldsymbol\theta,\sigma_s\rb + \frac{Q(\sigma_s)}{Q(\sigma_s+ib)} H_F^+\lb b,\boldsymbol\theta,\sigma_s\rb e^{ib\partial_{\sigma_s}}+ \frac{Q(\sigma_s)}{Q(\sigma_s-ib)}H_F^+\lb b,\boldsymbol\theta,-\sigma_s\rb e^{-ib\partial_{\sigma_s}},
\end{split}\eeq
where $H_F^+$ is given in \eqref{K}. Using the identity \eqref{propgb} satisfied by the $g_b$-function, it is straightforward to verify that the following identities hold for $k=\pm 1$:
\begin{align} \nonumber
\frac{Q(\sigma_s)}{Q(\sigma_s+kib)} = &\; \frac{\Gamma \left(1-b^2+2 i b k \sigma_s\right) \Gamma (1+2 ib k \sigma_s)}{\Gamma \left(1+b^2-2 i b k \sigma_s \right) \Gamma \left(1+2 b^2-2 i b k \sigma_s\right)} 
	\\ \label{identityN}
& \times \prod _{\epsilon =\pm} \frac{\Gamma \left(\frac{b Q}{2}+ib \left(\epsilon \theta_0+\theta_t-k \sigma_s\right)\right) \Gamma \left(\frac{b Q}{2}+ib \left(\epsilon \theta_\infty+\theta_1-k \sigma_s\right)\right)}{\Gamma \left(1-\frac{b Q}{2}+ib \left(\epsilon \theta_0+\theta_t+k \sigma_s\right)\right) \Gamma \left(1-\frac{b Q}{2}+ib\left(\epsilon \theta_\infty+\theta_1+k \sigma_s\right)\right)}.
\end{align}
We now substitute \eqref{identityN} into \eqref{5p18} and we use the identity $\Gamma(1-z)\Gamma(z)=\frac{\pi}{\operatorname{sin}{(\pi z)}}$. For $k=\pm1$ we obtain
\beq
\frac{Q(\sigma_s)}{Q(\sigma_s+kib)} H_F^+\lb b,\boldsymbol\theta,k\sigma_s\rb = C(b,b^{-1},L\boldsymbol{\theta},-k\sigma_s),
\eeq
where $C$ is defined by \eqref{fonctionC}. It finally remains to use the identification \eqref{VH0} of the potentials $V$ and $H_F^0$ to obtain \eqref{htildehR1}. \end{proof}

We have showed that the renormalized Virasoro fusion kernel \eqref{tildeF} and the renormalized Ruijsenaars hypergeometric function \eqref{fonctionR} are joint eigenfunctions of the same four difference operators. Therefore the two functions are proportional. We are now going to show that they are actually equal.

\subsection{$F_{\text{ren}}=R_{\text{ren}}$}
We now present the main result of this article.
\begin{theorem}\label{theoremRF}
The renormalized fusion kernel $F_{\text{ren}}$ given in \eqref{tildeF} and the renormalized Ruijsenaars hypergeometric function defined by \eqref{fonctionR} are equal:
\beq \label{RF}
F_{\text{ren}}\lb b,\boldsymbol\theta,\sigma_s,\sigma_t\rb = R_{\text{ren}}(b,b^{-1},L\boldsymbol\theta,\sigma_s,\sigma_t),
\eeq
where the matrix $L$ is given in \eqref{L}.
\end{theorem}

\begin{proof}
The equality is not immediate and follows from the identity \eqref{identitybarnes} satisfied by the hyperbolic Barnes integral. Let us first rewrite $F_{\text{ren}}$ using \eqref{tildeF} and the representation \eqref{fusion01}. It will also be convenient to perform a shift $x\to x-\tfrac{iQ}2$ in the integrand. This amounts to lifting up all poles of the integrand by $\tfrac{iQ}2$. This also maps the contour $\mathsf{F}$ to a contour $\mathsf{F}'$ which runs from $-\infty$ and $+\infty$ and lies in the strip $\im x \in ]0,\tfrac{iQ}2[$. It is now straightforward to verify that $F_{\text{ren}}$ can be written as
\beq\label{tildeF2}\begin{split}
F_{\text{ren}}\lb b,\boldsymbol\theta,\sigma_s,\sigma_t\rb = & \prod_{\epsilon_1=\pm} \lb \frac{s_b(\epsilon_1\sigma_t-\theta_0-\theta_\infty)}{s_b(\epsilon_1\sigma_s+\theta_1-\theta_\infty)}\prod_{\epsilon_2=\pm} s_b(\epsilon_1\sigma_s+\epsilon_2\theta_0-\theta_t) \rb \\
& \times \int_{\mathsf{F}'} dx ~ \prod_{\epsilon=\pm} \frac{s_b(x-\tfrac{iQ}2+\epsilon \theta_1)s_b(x-\tfrac{iQ}2+\theta_\infty+\theta_t+\epsilon\theta_0)}{s_b(x+\theta_\infty+\epsilon\sigma_s)s_b(x+\theta_t+\epsilon\sigma_t)}.
 \end{split}\eeq
Recalling the definition \eqref{hyperbolicbarnes} of the hyperbolic Barnes integral and using the identity $s_b(z)=G(b,b^{-1},z)$, the contour integral in \eqref{tildeF2} takes the form
\beq
\int_{\mathsf{F}'} dx ~ \prod_{\epsilon=\pm} \frac{s_b(x-\tfrac{iQ}2+\epsilon \theta_1)s_b(x-\tfrac{iQ}2+\theta_\infty+\theta_t+\epsilon\theta_0)}{s_b(x+\theta_\infty+\epsilon\sigma_s)s_b(x+\theta_t+\epsilon\sigma_t)} = \frac12 \mathcal{B}_h(b,b^{-1},\boldsymbol{v}).
\eeq
Several choices of $\boldsymbol{v}\in \mathcal{G}_{iQ}$ lead to the same contour integral. Here we choose 
\beq \begin{split}
& v_1= \theta_t-\sigma_t, \quad v_2=\theta_t+\sigma_t, \quad v_3=\tfrac{i Q}{2}+\theta_1, \quad v_4=\tfrac{i Q}{2}-\theta_0-\theta_\infty-\theta_t, \\
& v_5=\tfrac{i Q}{2}-\theta_1, \quad v_6=\tfrac{i Q}{2}+\theta_0-\theta_\infty-\theta_t, \quad v_7=\theta_\infty-\sigma_s, \quad v_8=\theta_\infty+\sigma_s.\end{split}\eeq
From \eqref{actionomega} we have $s=\frac{\theta_0-\theta_1+\theta_\infty-\theta_t}2$, and the action of $\omega$ on the parameters $\boldsymbol{v}\in \mathcal{G}_{iQ}$ yields
\beq
\omega \boldsymbol v = (v_1+s,v_2+s,...,v_5-s,...,v_8-s).
\eeq
%
We now apply the identity \eqref{identitybarnes}. A straightforward computation leads to \beq\label{tildeF3}
F_{\text{ren}}\lb b,\boldsymbol\theta,\sigma_s,\sigma_t\rb = \frac12 \prod_{\epsilon=\pm} \frac{s_b(\epsilon\sigma_s-\theta_0-\theta_t)}{s_b(\epsilon\sigma_t+\theta_1+\theta_t)} ~ \mathcal{B}_h(b,b^{-1},\omega \boldsymbol{v}),
\eeq
where 
\beq\label{tildeF4}\begin{split}
\mathcal{B}_h(b,b^{-1},\omega \boldsymbol{v}) = 2 \int_{\mathsf{F}'} & dx ~ s_b\left(x-\tfrac{iQ}2-\tfrac{\theta_0}{2}-\tfrac{\theta_1}{2}+\tfrac{3 \theta_\infty}{2}+\tfrac{\theta_t}{2}\right) s_b\left(x-\tfrac{iQ}2+\tfrac{\theta_0}{2}+\tfrac{\theta_1}{2}+\tfrac{\theta_\infty}{2}+\tfrac{3 \theta_t}{2}\right) \\
& \times \prod _{\epsilon=\pm} \frac{s_b\left(x-\tfrac{iQ}2+\epsilon \left(\frac{\theta_0}{2}+\frac{\theta_1}{2}+\frac{\theta_\infty}{2}-\frac{\theta_t}{2}\right)\right)}{s_b\left(x-\frac{\theta_0}{2}+\frac{\theta_1}{2}+\frac{\theta_\infty}{2}+\frac{\theta_t}{2}+\epsilon \sigma_s\right) s_b\left(x+\frac{\theta_0}{2}-\frac{\theta_1}{2}+\frac{\theta_\infty}{2}+\frac{\theta_t}{2}+\epsilon \sigma_t\right)}.
\end{split}\eeq
On the other hand, the function $R_\text{ren}$ defined in \eqref{fonctionR} satisfies
\beq\label{fonctionR2}\begin{split}
R_{\text{ren}}(b,&b^{-1},L\boldsymbol\theta,\sigma_s,\sigma_t) = \prod_{\epsilon=\pm1} \frac{s_b(\epsilon\sigma_s-\theta_0-\theta_t)}{s_b(\epsilon\sigma_t+\theta_1+\theta_t)} \\
 & \times \int_{\mathsf{R}}dz~ \frac{1}{s_b(z+\tfrac{i Q}{2}) s_b(z+2 \theta_t+\tfrac{i Q}{2})}\prod_{\epsilon=\pm1}\frac{s_b(z+\theta_0+\theta_t+\epsilon\sigma_s) s_b(z+\theta_1+\theta_t+\epsilon\sigma_t)}{s_b\left(z+\theta_0+\theta_1+\epsilon\theta_\infty+\theta_t+\tfrac{i Q}{2}\right)},
 \end{split}\eeq
where the contour $\mathsf{R}$ runs from $-\infty$ to $+\infty$, lying in the strip $\im z \in ]-\tfrac{iQ}2,0[$. 

The prefactors in \eqref{tildeF3} and \eqref{fonctionR2} coincide. Therefore it remains to identify the contour integral in \eqref{tildeF4} with the one in \eqref{fonctionR2}. Performing in \eqref{tildeF4} the change of variable $x=-z$ which maps the contour $\mathsf{F}'$ to $\mathsf{R}$ and using $s_b(z)=s_b(-z)^{-1}$, we obtain
\beq\label{tildeF5}\begin{split}
\mathcal{B}_h(b,b^{-1},\omega \boldsymbol{v}) = 2 \int_{\mathsf{R}} & dz ~ s_b\left(z+\tfrac{iQ}2+\tfrac{\theta_0}{2}+\tfrac{\theta_1}{2}-\tfrac{3 \theta_\infty}{2}-\tfrac{\theta_t}{2}\right)^{-1} s_b\left(z+\tfrac{iQ}2-\tfrac{\theta_0}{2}-\tfrac{\theta_1}{2}-\tfrac{\theta_\infty}{2}-\tfrac{3 \theta_t}{2}\right)^{-1} \\
& \times \prod _{\epsilon=\pm} \frac{s_b\left(z+\frac{\theta_0}{2}-\frac{\theta_1}{2}-\frac{\theta_\infty}{2}-\frac{\theta_t}{2}+\epsilon \sigma_s\right) s_b\left(z-\frac{\theta_0}{2}+\frac{\theta_1}{2}-\frac{\theta_\infty}{2}-\frac{\theta_t}{2}-\epsilon \sigma_t\right)}{s_b\left(z+\tfrac{iQ}2+\epsilon \left(\frac{\theta_0}{2}+\frac{\theta_1}{2}+\frac{\theta_\infty}{2}-\frac{\theta_t}{2}\right)\right)}.
\end{split}\eeq
It finally remains to perform $z\to z+\frac{\theta_1}2+\frac{\theta_0}2+\frac{\theta_\infty}2+\frac{3\theta_t}2$ in \eqref{tildeF5} to obtain \eqref{fonctionR2}.
\end{proof}

\section{Conclusion and perspectives}

In this article we have proved that the Virasoro fusion kernel is a joint eigenfunction of four difference operators. We have found a normalization of the conformal blocks for which the four difference operators are mapped to four versions of the quantum relativistic hyperbolic $BC_1$ Calogero-Moser Hamiltonian. We have consequently proved that the Virasoro fusion kernel and the Ruijsenaars hypergeometric function coincide up to normalization and are the quantum eigenfunction of this integrable system. We now mention a list of perspectives related to this work.

\begin{enumerate}
\item It would be interesting to understand the role played by the four-point Virasoro conformal blocks in the context of the present integrable system. In view of \eqref{s-t-tilde} and Theorem \ref{theoremRF}, the renormalized $R$-function can be seen as the kernel of the fusion transformation relating the s- and t-channel renormalized conformal blocks. On the other hand, in \cite{R2003} a unitary Hilbert space transform associated to the function $R_\text{ren}$ was constructed for special values of the couplings. We believe that this Hilbert space is in fact spanned by the four-point Virasoro conformal blocks.

\item The quantum eigenfunction of the $A_1$ relativistic hyperbolic CM system is a one-coupling specialization of the Ruijsenaars hypergeometric function \cite{R2011}. It would be interesting to compare this limit to the transition limit from the Virasoro fusion kernel to the Virasoro modular kernel \cite{HJS}. The latter describes how conformal blocks associated to the one-point torus transform under a mapping class group action.

\item A natural question is to find higher rank generalizations of our result. The $BC_N$, $N>1$ generalization of the function $R_\text{ren}$ have not yet been constructed. Higher rank analogs of the Virasoro fusion kernel are associated to $\mathcal{W}$-type algebras and also remain to be found. However, the framework developed in \cite{SS1,SS2} allows us, in principle, to construct such a generalization from a quantum group perspective.

\item What is the meaning of the classical and/or non-relativistic limits in the conformal blocks setting?

\item The $R$-function reduces to the celebrated Askey-Wilson polynomials in a limit where one of $v, \hat{v}$ is discretized \cite{R1999}. What does this limit mean from the conformal blocks viewpoint?

\item The Askey-Wilson polynomials can be studied using representation theory of the double affine hecke algebra and the Askey-Wilson algebra \cite{NS,KM}. An interesting program would be to generalize this algebraic study to the case of $R_\text{ren}$. We believe that Virasoro conformal blocks and their quantum monodromies provide the correct framework for such a study.

\item Various confluent limits of the Virasoro fusion kernel were constructed in \cite{LR}. It would be interesting to understand these limits from an integrable system point of view.

\end{enumerate}

\appendix
\section{Special functions} \label{appendixA}
Equation \eqref{fonctionR} expresses the renormalized Ruijsenaars hypergeometric function in terms of the hyperbolic gamma function $G(a_-,a_+,z)=E(a_-,a_+,z)/E(a_-,a_+,-z)$. Let us omit the dependence of $G$ and $E$ on $a_-$ and $a_+$ for simplicity. Following \cite{R1999}, the two functions are defined by
\begin{equation}\label{defsb}
G(z)=\operatorname{exp}{\left[i \int_0^\infty \frac{dy}{y} \left(\frac{\operatorname{sin}{2yz}}{2\operatorname{sinh}{(a_+y)}\operatorname{sinh}{(a_-y)}}-\frac{z}{a_+ a_- y}\right)\right]}, \qquad |\im z|<\frac{a_++a_-}{2},
\end{equation}
and
\beq \label{gb}
E(z)=\operatorname{exp}\left[\frac{1}4 \int_0^\infty \frac{dy}{y}\left(\frac{1-e^{-2iyz}}{\operatorname{sinh}{(a_+ y)} \operatorname{sinh}{(a_- y)}}-\frac{2iz}{a_+a_- y} - \frac{z^2}{a_+ a_-}(e^{-2a_+y}+e^{-2a_-y})\right)\right].
\eeq
The functions $G$ and $E$ are obviously invariant under the exchange of $a_-$ and $a_+$. Moreover, $G$ satisfies the difference equations
\beq
\frac{G(z+ia_-/2)}{G(z-ia_-/2)}=2\operatorname{cosh}{(\pi z/a_+)}, \qquad \frac{G(z+ia_+/2)}{G(z-ia_+/2)}=2\operatorname{cosh}{(\pi z/a_-)}.
\eeq
As long as the ratio $a_-/a_+$ stays in the subset $\mathbb{C} \setminus (-\infty, 0]$, the function $G$ extends to a meromorphic function with poles and zeros given by
\beq\label{polesGhyp}
z=-i(k+\tfrac12)a_--i(l+\tfrac12)a_+, \qquad k,l \in \mathbb{N}, \qquad (\text{poles}), 
\eeq
\beq\label{zeroGhyp}
z=i(k+\tfrac12)a_-+i(l+\tfrac12)a_+, \qquad k,l \in \mathbb{N}, \qquad (\text{zeros}).
\eeq 
Finally the function $G$ is scale invariant:
\beq\label{scale}
G(\lambda a_-, \lambda a_+, \lambda z)=G(a_-,a_,z), \quad \lambda > 0.
\eeq
The function $E$ satisfies the difference equation
\beq\label{A6}
\frac{E(z+ia_-/2)}{E(z-ia_-/2)}=\frac{\sqrt{2\pi}}{\Gamma(\tfrac12+\frac{iz}{a_+})} \lb \tfrac{a_-}{a_+}\rb^{\tfrac{iz}{2a_+}},
\eeq
and the difference equation obtained by exchanging $a_-$ and $a_+$ in \eqref{A6}. The function $E$ is holomorphic as long as $a_-/a_+ \in \mathbb{C} \setminus (-\infty, 0]$. Moreover, $E$ has zeros located at
\beq
z=i(k+\tfrac12)a_-+i(l+\tfrac12)a_+, \qquad k,l \in \mathbb{N}.
\eeq

On the other hand, the Virasoro fusion kernel \eqref{fusion01} is defined in terms of two special functions $s_b(z)$ and $g_b(z)$ such that $s_b(z)=g_b(z)/g_b(-z)$. They are related to the functions $G$ and $E$ as follows:
\begin{align}\label{sbgbGE}
s_b(z) = G(b, b^{-1}; z), \qquad g_b(z) = \frac{1}{E(b, b^{-1}; -z)}.
\end{align}
The properties of $s_b$ and $g_b$ can be deduced from the ones of $G$ and $E$. Both $s_b$ and  $g_b$ are obviously invariant under the exchange of $b$ and $b^{-1}$. Most importantly, the function $g_b(z)$ satisfies the difference equations
\beq \label{propgb}
\frac{g_b \lb z+\frac{ib}{2}\rb}{g_b\lb z-\frac{ib}{2}\rb}=\frac{b^{-ibz}\sqrt{2\pi}}{\Gamma \lb \frac{1}{2}-ibz \rb}, \qquad
\frac{g_b \lb z+\frac{i}{2b}\rb}{g_b\lb z-\frac{i}{2b}\rb}=\frac{b^{-\frac{iz}{b}}\sqrt{2\pi}}{\Gamma \lb \frac{1}{2}- \frac{iz}{b} \rb}.
\eeq
Finally, the $s_b$ function satisfies 
\begin{equation}\label{differencesb}
\frac{s_b(z+\frac{ib}{2})}{s_b(z-\frac{ib}{2})}=2\operatorname{cosh}{\pi b z}, \qquad \frac{s_b(z+\frac{i}{2b})}{s_b(z-\frac{i}{2b})}=2\operatorname{cosh}{\frac{\pi z}{b}}.
\end{equation}

\section{Hyperbolic Barnes integral}\label{appendixB}

In this appendix we define the hyperbolic Barnes integral $\mathcal B_h(a_-,a_+,\boldsymbol{u})$ following \cite{BRS}. We describe one of its symmetry properties which is the key to the proofs of Proposition \ref{lemma1} and Theorem \ref{theoremRF}. 

Define the complex hyperplane $\mathcal{G}_k$ such that
\begin{equation}\label{hyperplane}
\mathcal{G}_k = \{ u=(u_1,...u_8) \in \mathbb{C}^8 | \sum_{j=1}^8 u_j=2k \}.
\end{equation}
The hyperbolic Barnes integral $\mathcal B_h(a_-,a_+,\boldsymbol{u})$ is defined for $\boldsymbol{u} \in \mathcal{G}_{2ia}$ and $a=\tfrac{a_-+a_+}2$ by
\begin{equation}\label{hyperbolicbarnes}
\mathcal B_h(a_-,a_+,\boldsymbol{u}) = 2 \int_{\mathcal{C}} \frac{\prod_{j=3}^6 G\left(z-u_j\right)}{\prod_{j=1,2,7,8}^6 G\left(z+u_j\right)}dz,
\end{equation}
where the contour $\mathcal{C}$ goes from $-\infty$ and $+\infty$, passing in between the upper and lower sequences of poles of the integrand. 
\begin{lemma}\cite[Proposition 4.7]{BRS} 
Let $\omega$ be an operator acting on $\boldsymbol u \in \mathcal{G}_{2ia}$ as
\begin{equation}\label{actionomega}
\omega \boldsymbol{u} = (u_1+s,...,u_4+s,u_5-s,...,u_8-s), \qquad s=ia-\frac{1}{2}(u_1+u_2+u_3+u_4).
\end{equation}
The hyperbolic Barnes integral $B_h(a_-,a_+,\boldsymbol{u})$, $\boldsymbol{u} \in \mathcal{G}_{2ia}$ is invariant under permutations of $(u_1, u_2, u_7,u_8)$ and of $(u_3,u_4,u_5,u_6)$, and it satisfies\footnote{A similar identity using a degeneration limit of the Spiridonov's $V$-function was obtained in \cite{TV14}.}
\begin{equation}\label{identitybarnes}
\mathcal B_h(a_-,a_+,\boldsymbol{u})=\mathcal B_h(a_-,a_+,\omega \boldsymbol{u})\displaystyle\prod_{j=1,2}\prod_{k=3,4}G\left(i a-u_j-u_k\right)\displaystyle\prod_{j=5,6}\prod_{k=7,8}G\left(ia-u_j-u_k\right).
\end{equation}
\end{lemma}

\bigskip
\noindent
{\bf Acknowledgement} {\it The author would like to thank P. Baseilhac, E. Langmann, S. Ribault and S. Ruijsenaars for valuable discussions and encouragements. The author is also grateful to O. Lisovyy for many helpful discussions and for his comments on a first version of the manuscript. Finally, the author is indebted to J. Lenells for his encouragements and his careful checks on the latest version of this paper. Support is acknowledged from the European Research Council, Grant Agreement No. 682537.}

\end{document}